\documentclass[a4paper,UKenglish,cleveref, autoref, thm-restate,authorcolumns]{lipics-v2019}

\input{preamble.tex}

\title{Computation of Hadwiger Number and Related Contraction Problems: Tight Lower Bounds} 

\titlerunning{Tight Lower Bounds on the Computation of Hadwiger Number} 

\author{Fedor V. Fomin}{University of Bergen, Bergen, Norway, fomin@ii.uib.no}{}{}{Research Council of Norway via the project MULTIVAL.}

\author{Daniel Lokshtanov}{University of California, Santa Barbara, USA, daniello@ucsb.edu}{}{}{ European Research Council (ERC) under the European Union’s Horizon
2020 research and innovation programme (grant no. 715744), and United States - Israel Binational Science Foundation grant no. 2018302.\begin{minipage}{0.1\textwidth}
    \begin{center}
        \includegraphics[scale=0.5]{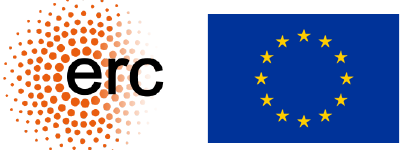}
    \end{center}
\end{minipage}}

\author{Ivan Mihajlin}{University of California, San Diego, California, USA}{imikhail@cs.ucsd.edu}{}{}

\author{Saket Saurabh}{Department of Informatics, University of Bergen, Norway, and The Institute of Mathematical Sciences, HBNI and IRL 2000 ReLaX, Chennai, India, saket@imsc.res.in}{}{}{European Research Council (ERC) under the European Union’s Horizon 2020 research and innovation programme (grant no. 819416), and Swarnajayanti Fellowship
grant DST/SJF/MSA-01/2017-18.\begin{minipage}{0.1\textwidth}
    \begin{center}
        \includegraphics[scale=0.5]{ERCandEU.pdf}
    \end{center}
\end{minipage}}

\author{Meirav Zehavi}{Ben-Gurion University of the Negev, {Beer-Sheva, Israel}, meiravze@bgu.ac.il}{}{}{Israel Science Foundation grant no. 1176/18, and United States – Israel Binational Science Foundation grant no. 2018302.}

\authorrunning{F.V. Fomin, D. Lokshtanov, I. Mihajlin, S. Saurabh and M. Zehavi} 

\Copyright{Fedor V. Fomin, Daniel Lokshtanov, Ivan Mihajlin, Saket Saurabh and Meirav Zehavi} 



\EventEditors{John Q. Open and Joan R. Access}
\EventNoEds{2}
\EventLongTitle{42nd Conference on Very Important Topics (CVIT 2016)}
\EventShortTitle{CVIT 2016}
\EventAcronym{CVIT}
\EventYear{2016}
\EventDate{December 24--27, 2016}
\EventLocation{Little Whinging, United Kingdom}
\EventLogo{}
\SeriesVolume{42}
\ArticleNo{23}

\begin{document}

\maketitle

\begin{abstract}
We prove that  the Hadwiger number of an $n$-vertex graph $G$ (the maximum size of a clique minor in $G$) cannot be computed in time $n^{o(n)}$, unless the Exponential Time Hypothesis (ETH) fails. This resolves a well-known open question in the area of exact exponential algorithms. 
The technique developed for resolving  the Hadwiger number problem has a wider applicability. We use it  to rule out the existence of $n^{o(n)}$-time algorithms (up to ETH) for  a large class of computational problems  concerning edge contractions in graphs. 
\end{abstract}

\newpage

  \newcommand{\GM}{\textsc{Graph Minor}\xspace}    
    \newcommand{\TGM}{\textsc{Topological Graph Minor}\xspace}  
    \newcommand{\cOs}{\mathcal{O}^*}
\section{Introduction}\label{sec:intro}

The \emph{Hadwiger number} $h(G)$ of a  graph $G$ is  the largest number $h$ for which the complete graph $K_h$ is a minor of $G$. Equivalently, $h(G)$ is 
the maximum  size of the largest complete graph that can be obtained from $G$  by contracting edges. It is named after Hugo Hadwiger, who conjectured    in 1943 
that the Hadwiger number of $G$ is always at least as large as its  chromatic number.
According to Bollob\'{a}s,   Catlin, and Erd\H{o}s, this conjecture remains  ``one of the deepest unsolved problems in graph theory'' \cite{MR593989}.

The Hadwiger number of an $n$-vertex graph $G$ can be easily computed in time $n^{\cO(n)}$ by  brute-forcing through all possible partitions  of the vertex set of $G$ into connected sets, contracting each set into one vertex and checking whether the resulting graph is a complete graph.   The question whether 
the Hadwiger  number of a graph can be computed in   single-exponential  $2^{\cO(n)}$ time 
was previously asked in  \cite{agrawal2017split,DBLP:journals/jacm/CyganFGKMPS17,LingasW09}. 
Our main result  provides a negative answer to this  open question. 

\begin{theorem}\label{thm:informalmain} Unless the Exponential Time Hypothesis (ETH) is false, there does not exist an algorithm computing the Hadwiger number of an $n$-vertex graph in time $n^{o(n)}$. 
\end{theorem}

%
%
%
%

 The  interest in the complexity of the Hadwiger number  is naturally explained by  the recent developments in the area of exact exponential algorithms, that is, algorithms 
solving intractable problems significantly faster than the trivial exhaustive search, though still in exponential time \cite{Fomin:2010mo}. Within the last decade, significant progress  on  upper and lower bounds of exponential algorithms has been achieved.  Drastic improvements over  brute-force algorithms were obtained for a number of fundamental problems like \textsc{Graph Coloring}~\cite{BjorklundHK2009-Se} and  \textsc{Hamiltonicity}~\cite{Bjorklund10}.
  On the other hand, by making use of the ETH, lower bounds could be obtained for 
 \textsc{2-CSP} \cite{T2008} or for \textsc{Subgraph Isomorphism} and \textsc{Graph Homomorphism}~\cite{DBLP:journals/jacm/CyganFGKMPS17}. 
 
 \medskip
 \GM  (deciding whether a graph $G$ contains a graph $H$ as a minor) is a fundamental  problem in graph theory and graph algorithms.   \GM could be seen as special case of a general  graph embedding problem where one wants to embed  a graph $H$ into graph $G$. In what follows we will use $n$ to denote the number of vertices in $G$ and $h$ to denote the number of vertices in $H$. 
 By the  theorem of 
 Robertson and Seymour 
\cite{RobertsonS-GMXIII},  there exists a computable function $f$ and an algorithm that, for  given graphs $G$ and $H$,
checks in time  $f(h) \cdot n^3$   whether $H$ is a minor of $G$.  Thus the problem is fixed-parameter tractable (FPT) being parameterized by $H$.   On the other hand, Cygan et al.  \cite{DBLP:journals/jacm/CyganFGKMPS17} proved that   unless the ETH fails,  this problem cannot be solved in time $n^{o(n)}$ even  in the case when  $|V(G)|=|V(H)|$.  
 Other interesting embedding problems that are strongly related to \GM include the following problems.  
 \begin{itemize}
 \item \textsc{Subgraph Isomorphism}: Given two graphs $G$ and $H$, decide whether $G$ contains a subgraph isomorphic to $H$.  This problem cannot be solved in time $n^{o(n)}$ when $|V(G)|=|V(H)|$, unless the ETH fails~\cite{DBLP:journals/jacm/CyganFGKMPS17}.   
 In the special case  called \textsc{Clique},   when $H$ is a clique, a brute-force algorithm checking for every vertex subset of $G$ whether it is a clique of size $h$ solves the problem in time $n^{\cO(h)}$. The same algorithm also runs in single-exponential time $\cO(2^n n^2)$.
 It is also known that  \textsc{Clique}  is W[1]-hard parameterized by $h$ and cannot be solved in time 
  $f(h) \cdot n^{o(h)}$ for any function $f$ unless the ETH fails~\cite{DowneyF99,Chen2005216}. 
 
  \item \textsc{Graph Homomorphism}: Given two graphs $G$ and $H$, decide whether there exists a homomorphism from $G$ to $H$. (A {\em homomorphism} $G\to H$ from an undirected graph $G$ to an undirected graph $H$ is a mapping
from the vertex set of $G$ to that of $H$ such that the image of every edge of $G$ is an edge of $H$.) 
 This problem is trivially solvable in time $h^{\cO(n)}$, and  an algorithm of running time $h^{o(n)}$ for this problem would yield the failure of the ETH \cite{DBLP:journals/jacm/CyganFGKMPS17}. However, for the special case of $H$ being a clique,   \textsc{Graph Homomorphism} is equivalent to  \textsc{$h$-Coloring} (deciding whether the chromatic number of $G$ is at most $h$), and thus is solvable in single-exponential time $2^{n} \cdot n^{\cO(1)}$~\cite{BjorklundHK2009-Se,Lawler76}.  When the graph $G$ is a complete graph,  the problem is equivalent to finding a clique of size $n$ in $H$, and then is solvable in time $2^{h} \cdot h^{\cO(1)}$.
       \item \textsc{Topological Graph Minor}: Given two graphs $G$ and $H$, decide whether $G$ contains $H$ as a topological minor. (We say that a graph $H$ is a subdivision of a graph $G$ if $G$ can be obtained from $H$ by contracting only edges incident with at least one vertex of degree two. 
A~graph $H$ is called a \emph{topological minor} of a graph $G$  if a subdivision of $H$ is isomorphic to a subgraph of $G$.)
This problem is, perhaps, the closest ``relative'' of \GM.   Grohe et al.~\cite{DBLP:conf/stoc/GroheKMW11} gave an algorithm of  running time  $f(h) \cdot n^3$ for this problem for some 
  computable function $f$.   
Similar to \GM and \textsc{Subgraph Isomorphism}, this problem cannot be solved in time $n^{o(n)}$ when $|V(G)|=|V(H)|$, unless the ETH fails \cite{DBLP:journals/jacm/CyganFGKMPS17}.   However for the special case of the problem with $H$ being a complete graph, Lingas and Wahlen 
\cite{LingasW09} gave a single-exponential algorithm solving the problem in  time $2^{\cO(n)}$.

\end{itemize}
 
 Thus all the above graph embedding ``relatives'' of \GM are solvable in single-exponential time  when graph $H$ is a clique.  However, from the perspective of exact exponential algorithms,  Theorem~\ref{thm:informalmain} implies that 
finding the largest  clique minor is the most 
  difficult problem out of them all. This is why we find the lower bound provided by Theorem~\ref{thm:informalmain}   surprising. 
 Moreover, from the  perspective of parameterized complexity, finding a clique minor of size $h$, 
which is FPT, is actually easier than  finding a clique (as a subgraph) of size $h$, which is W[1]-hard,  as well as from   finding an {$h$-coloring} of a graph, which is para-NP-hard.

Theorem~\ref{thm:informalmain} also answers another  question of 
Cygan et al.~\cite{DBLP:journals/jacm/CyganFGKMPS17}, who asked 
whether 
deciding if a graph $H$  can be obtained from a graph $G$ only by edge contractions,  could be resolved in single-exponential time.  By Theorem~\ref{thm:informalmain}, the existence of such an algorithm is highly unlikely even when the graph $H$ is a complete graph. Moreover, the technique developed to prove Theorem~\ref{thm:informalmain}, appears to be extremely useful to rule out the existence of  $n^{o(n)}$-time algorithms for various contraction problems. We formalize our results with the following {\sc $\cal F$-Contraction} problem. Let  $\cal F$ be a graph class. 
Given a graph $G$ and $t\in\mathbb{N}$, the task is to decide whether there exists a subset $F\subseteq E(G)$ of size at most $t$ such that $G/F\in {\cal F}$ (where $G/F$ is the graph obtained from $G$ by contracting the edges in $F$). We prove that in each of  the cases of {\sc $\cal F$-Contraction} where $\cal F$ is the family of chordal graphs, interval graphs, proper interval graphs, threshold graphs, trivially perfect graphs, split graphs, complete split graphs and perfect graphs, unless the ETH fails,  {\sc $\cal F$-Contraction} is not solvable in time  $n^{o(n)}$. For lack of space, these results are relegated to Appendix~\ref{sec:contractionToClasses}.

\begin{figure}[t]
 \center{\includegraphics[scale=0.6]{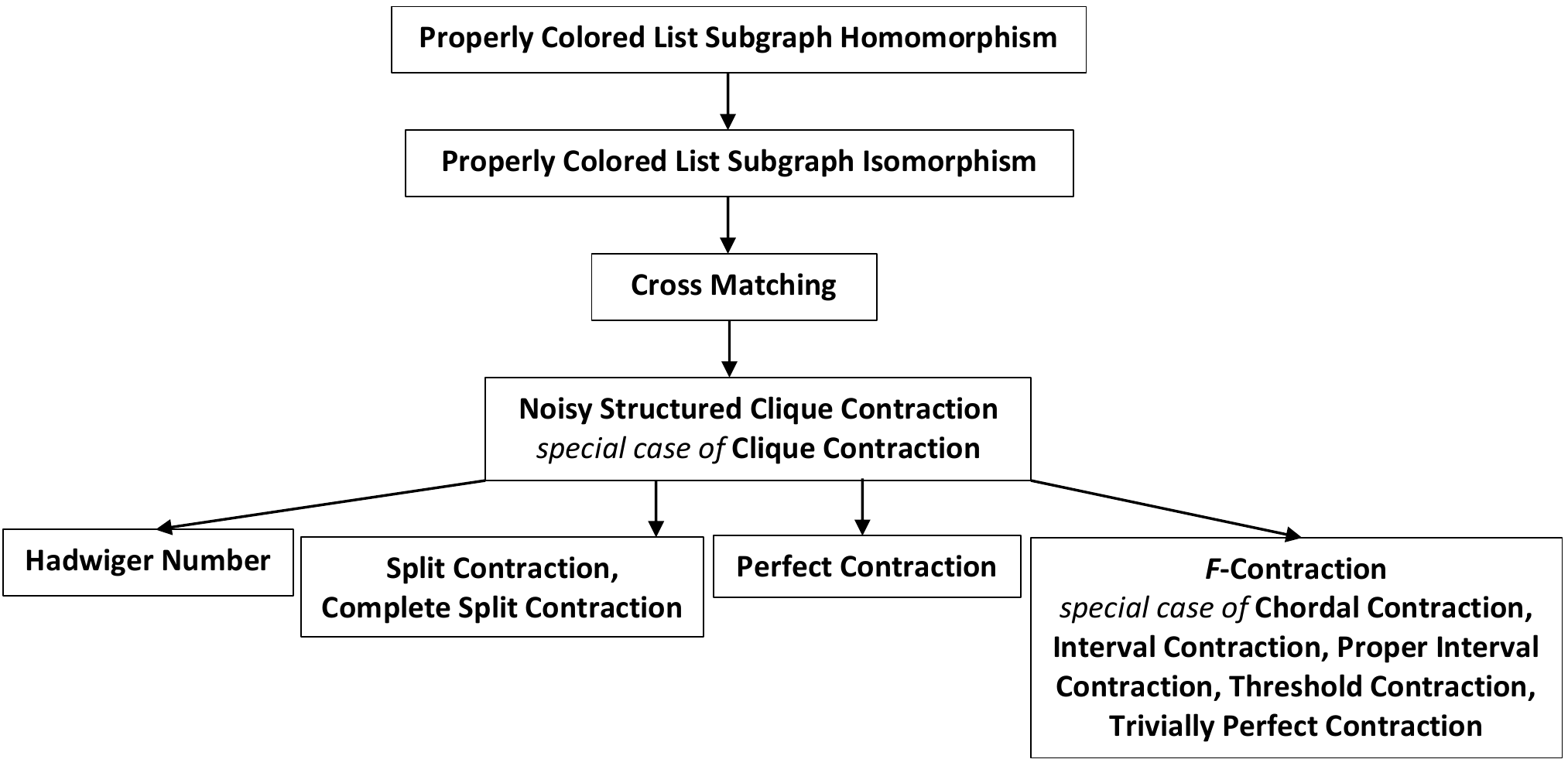}}
\caption{\label{fig:outline} A summary of the problems considered in this paper, and the reductions between them.}
\end{figure}

\subparagraph{Technical Details.} A summary of the reductions presented in this paper is given in Fig.~\ref{fig:outline}.  To prove our lower bounds, we first revisit the proof  of 
Cygan et al.~\cite{DBLP:journals/jacm/CyganFGKMPS17} for the ETH-hardness of a problem called {\sc List Subgraph Isomorphism}. Informally, in this problem we are given two graphs $G$ and $H$ on the same number of vertices, as well as a list of vertices in $H$ for each vertex in $G$, and we need to find a copy of $G$ in $H$ so that each vertex $u$ in $G$ is mapped to a vertex $v$ in $H$ that belongs to its list (i.e.~$v$ belongs to the list of $u$). We prove that the instances produced by the reduction (after some modification) of~\cite{DBLP:journals/jacm/CyganFGKMPS17}  have a very useful property that we crucially exploit later. Specifically, we construct a proper coloring of $G$ as well as a proper coloring of $H$, and show that every vertex $v$ in $H$ that belongs to the list of some vertex $u$ is, in fact, of the same color as $u$.

Having proved the above, we turn to prove the ETH-hardness of a special case of {\sc Clique Contraction} where the input graph is highly structured. To this end, we introduce an intermediate problem called {\sc Cross Matching}. Informally, in this problem we are given a graph $L$ with a partition $(A,B)$ of its vertex set, and need to find a perfect matching between $A$ and $B$ whose contraction gives a clique. To see the connection between this problem and {\sc List Subgraph Isomorphism}, think of the subgraph of $L$ induced by one side of the partition---say, $A$---as a representation of the {\em complement} of $G$, and the subgraph of $L$ induced by the other side of the partition as a representation of $H$. Then, the edges that go across $A$ and $B$ in a perfect matching can be thought of as a mapping of the vertices of $G$ to the vertices of $H$.  The crossing edges of $L$ are easily defined such that necessarily a vertex of $G$ can only be matched to a vertex in its list. In particular, we would like to enforce that every ``non-edge'' of the complement of $G$ (which corresponds to an edge of $G$) would have to be mapped to an edge of $H$ in order to obtain a clique.  However, the troublesome part is that non-edges of the complement of $G$ may also be ``filled'' (to eventually get a clique) using crossing edges rather than only edges of $H$.  To argue that this critical issue does not arise, we crucially rely on the proper colorings of~$G$~and~$H$.

Now, for the connection between {\sc Cross Matching} and {\sc Clique Contraction}, note that a solution to an instance of {\sc Cross Matching} is clearly a solution to the instance of {\sc Clique Contraction} defined by the same graph, but the other direction is not true. By adding certain vertices and edges to the graph of an instance of {\sc Cross Matching}, we enforce all solutions to be perfect matchings between $A$ and $B$. In particular, we construct the instances of {\sc Clique Contraction} in a highly structured manner that allows us to derive not only the ETH-hardness of {\sc Clique Contraction} itself, but to build upon them and further derive ETH-hardness for a wide variety of other contraction problems. In particular, we show that the addition of ``noise'' (that is, extra vertices and edges) to any structured instance of {\sc Clique Contraction} has very limited effect. Roughly speaking, we show that the edges in the ``noise'' and the edges going across the ``noise'' and core of the graph (that is, the original vertices corresponding to the structured instance of {\sc Clique Contraction}) are not ``helpful'' when trying to create a clique on the core (i.e.~it is not helpful to try to use these edges in order to fill non-edges between vertices in the core). Depending on the contraction problem at hand, the noise is slightly different, but the proof technique stays the same---first showing that the core must yield a clique, and then using the argument above (in fact, in all cases but that of perfect graphs, we are able to invoke the argument as a black box) to show that the noise is, in  a sense,~irrelevant.

\medskip
\noindent{\bf Preliminaries.} As we only use standard notations, we relegate them to Appendix~\ref{sec:prelims}.

\section{Lower Bound: {\sc Prop-Colored List Subgraph Isomorphism}}\label{sec:listSI}

In this section we  build upon the work of Cygan et al.~\cite{DBLP:journals/jacm/CyganFGKMPS17} and show a lower bound for a problem called {\sc Properly Colored List Subgraph Isomorphism (Prop-Col LSI)}. Intuitively, {\sc Prop-Col LSI} is a variant of {\sc Spanning Subgraph Isomorphism} where given two graphs $G$ and $H$, we ask whether $G$ is isomorphic to some spanning subgraph of $H$. The input to the variant consists also of proper colorings of $G$ and $H$ and an additional labeling of vertices in $G$ by subsets of vertices in $H$ of the same color, so that each vertex in $G$ can be mapped only to vertices in $H$ contained in its list.
Formally, it is defined as~follows.

\defproblem{{\sc Properly Colored List Subgraph Isomorphism (Prop-Col LSI)}}{Graphs $G$ and $H$ 
with proper colorings $c_G: V(G)\rightarrow \{1,\ldots,k\}$ and $c_H: V(H)\rightarrow \{1,\ldots,k\}$ for some $k\in\mathbb{N}$,  respectively, and a function $\ell: V(G)\rightarrow 2^{V(H)}$ such that for every $u\in V(G)$ and $v\in \ell(u)$, $c_G(u)=c_H(v)$.}{Does there exist a bijective function $\varphi: V(G)\rightarrow V(H)$ such that {\em (i)} for every $\{u,v\}\in E(G)$, $\{\varphi(u),\varphi(v)\}\in E(H)$, and {\em (ii)} for every $u\in V(G)$, $\varphi(u)\in \ell(u)$?}

Notice that as the function $\varphi$ above is bijective rather than only injective, we seek a spanning subgraph. Our objective is to prove the following statement.

\begin{lemma}\label{res:PropColLSI}
Unless the ETH is false, there does not exist an algorithm that solves {\sc Prop-Col LSI} in time $n^{o(n)}$ where $n=|V(G)|$.
\end{lemma}

In \cite{DBLP:journals/jacm/CyganFGKMPS17}, the authors considered the two problems defined below. Intuitively, the second is defined as {\sc Prop-Col LSI} when no proper colorings of $H$ and $G$ are given (and hence the labeling of vertices in $G$ is not restricted accordingly); the first is defined as the second when we seek a homomorphism rather than an isomorphism (i.e., the sought function $\varphi$ may not be injective) and also $|V(G)|$ may not be equal to $|V(H)|$ (thus $\varphi$ may neither be onto).

\defproblem{{\sc List Subgraph Homomorphism (LSH)}}{Graphs $G$ and $H$, and a function $\ell: V(G)\rightarrow 2^{V(H)}$ .}
{Does there exist a function $\varphi: V(G)\rightarrow V(H)$ such that {\em (i)} for every $\{u,v\}\in E(G)$, $\{\varphi(u),\varphi(v)\}\in E(H)$, and {\em (ii)} for every $u\in V(G)$, $\varphi(u)\in \ell(u)$?}

\defproblem{{\sc List Subgraph Isomorphism (LSI)}}{Graphs $G$ and $H$ where $|V(G)|=|V(H)|$, and a function $\ell: V(G)\rightarrow 2^{V(H)}$.}
{Does there exist a bijective function $\varphi: V(G)\rightarrow V(H)$ such that {\em (i)} for every $\{u,v\}\in E(G)$, $\{\varphi(u),\varphi(v)\}\in E(H)$, and {\em (ii)} for every $u\in V(G)$, $\varphi(u)\in \ell(u)$?}

The proof of hardness of {\sc LSI} consists of two parts:

\begin{itemize}
    \item Showing {\sc ETH}-hardness of {\sc LSH}.
    \item Giving a fine-grained reduction from {\sc LSH} to {\sc LSI}.
\end{itemize}

We cannot use the hardness of {\sc LSI} as a black box because {\sc Prop-Col LSI} is a special case of {\sc LSI}. Nevertheless, we will prove that the instances generated by the reduction (with a minor crucial modification) of Cygan et al.~\cite{DBLP:journals/jacm/CyganFGKMPS17} have the additional properties required to make them instances of our special case.

\subparagraph{Lower Bound: {\sc Properly Colored Subgraph Homomorphism}.}
Adapting the scheme of Cygan et al.~\cite{DBLP:journals/jacm/CyganFGKMPS17} to our purpose, we will first show that finding a homomorphism remains hard if it has to preserve a given proper coloring:

\defproblem{{\sc Properly Colored List Subgraph Homomorphism (Prop-Col LSH)}}{Graphs $G$ and $H$ with proper colorings $c_G: V(G)\rightarrow \{1,\ldots,k\}$ and $c_H: V(H)\rightarrow \{1,\ldots,k\}$ for some $k\in\mathbb{N}$,  respectively, and a function $\ell: V(G)\rightarrow 2^{V(H)}$ such that for every $u\in V(G)$ and $v\in \ell(u)$, $c_G(u)=c_H(v)$.}{Does there exist a function $\varphi: V(G)\rightarrow V(H)$ such that {\em (i)} for every $\{u,v\}\in E(G)$, $\{\varphi(u),\varphi(v)\}\in E(H)$, and {\em (ii)} for every $u\in V(G)$, $\varphi(u)\in \ell(u)$?}

In \cite{DBLP:journals/jacm/CyganFGKMPS17}, the authors gave a reduction from the {\sc $3$-Coloring} problem on $n$-vertex graphs of degree $4$ (which is known not to be solvable in time $2^{o(n)}$ unless the ETH fails), which generates equivalent instances $(G',H',\ell)$ of  {\sc LSH}  where both $|V(G')|$ and $|V(H')|$ are bounded by $\OO(\frac{n}{\log n})$. This proves that {\sc LSH} is not solvable in time $n^{o(n)}$ where $n=\max\{|V(G)|,|V(H)|\}$ unless the ETH fails. For their reduction, Cygan et al.~\cite{DBLP:journals/jacm/CyganFGKMPS17} considered the notion of a {\em grouping} (also known as {\em quotient graph}) $\widetilde{G}$ of a graph $G$ is a graph with vertex set $V(\widetilde{G})=\{B_1,B_2,\ldots,B_t\}$ where $(B_1,B_2,\ldots,B_t)$ is a partition of $V(G)$ for some $t\in\mathbb{N}$ and for any distinct $i,j\in\{1,\ldots,t\}$, the vertices $B_i$ and $B_j$ are adjacent in $\widetilde{G}$ if and only if there exist $u\in B_i$ and $v\in B_j$ that are adjacent in $G$. Specifically, they computed a grouping with a coloring having specific properties as stated in the following lemma (see also Fig.~\ref{fig:reduction} in Appendix \ref{app:detHomToISO}.).  

\begin{lemma}[Lemma 3.2 in \cite{DBLP:journals/jacm/CyganFGKMPS17}]\label{lem:coloring}
For any constant $d \geq 1$, there exist positive integers $\lambda = \lambda(d)$, $n_0 = n_0(d)$ and a polynomial time algorithm that for a given graph $G$ on $n \geq n_0$ vertices of maximum degree $d$ and a positive integer $r\leq \sqrt{\frac{n}{2\lambda}}$, finds a grouping $\widetilde{G}$ of $G$ and a coloring $\widetilde{c} : V(\widetilde{G}) \rightarrow [\lambda r]$ with the following properties: 
\begin{enumerate}
\item\label{lem:coloring1} $|V(\widetilde{G})| \leq |V(G)|/r$;
\item\label{lem:coloring2} The coloring $\widetilde{c}$ is a proper coloring of $\widetilde{G}^2$;\footnote{The square $G^2$ of a graph $G$ is the graph on vertex set $V(G)$ and edge set $\{\{u,v\}: \{u,v\}\in E(G)$ or there exists $w\in V(G)$ with $\{u,w\},\{v,w\}\in E(G)\}$.}
\item\label{lem:coloring3} Each vertex of $\widetilde{G}$ is an independent set in $G$;
\item\label{lem:coloring4} For any edge $\{B_i,B_j\}\in E(\widetilde{G})$, there exists exactly one pair $(u,v)\in B_i\times B_j$ such that $\{u,v\}\in E(G)$.
\end{enumerate}
\end{lemma}

Now, we describe the reduction of \cite{DBLP:journals/jacm/CyganFGKMPS17}. Here, without loss of generality, it is assumed that $G$ has no isolated vertices, else they can be removed. An explanation of the intuition behind this somewhat technical definition is given below it.

\begin{definition}\label{def:reduceLSH}
For any instance $G$ of {\sc 3-Coloring} where $G$ has degree $d$ and a positive integer $r=o(\sqrt{|V(G)|})$, the instance $\reduce(G)=(\widetilde{G},\widetilde{H},\ell)$ of {\sc LSH} is defined as follows.
\begin{itemize}
\item {\bf The graph $\widetilde{G}$.} Let $\widetilde{G}$ and $\widetilde{c}: V(\widetilde{G})\rightarrow \{1,2,\ldots,L\}$ be the grouping and coloring given by Lemma \ref{lem:coloring} where $L=\lambda(d)r$. Additionally, for each $B\in V(\widetilde{G})$, define $\phi_B: \{1,2,\ldots,L\} \rightarrow B\cup\{0\}$ as follows: for any $i\in\{1,2,\ldots,L\}$, if there exists $(u,v,B')$ such that $u\in B$ and $v\in B'$, $\{u,v\}\in E(G)$ and $\widetilde{c}(B')=i$, then $\phi_B(i)=u$, and otherwise $\phi_B(i)=0$.\footnote{The uniqueness of $u$ (if it exists), and thus the validity of $\phi_B$, follows from Properties \ref{lem:coloring2} and~\ref{lem:coloring4} in~Lemma~\ref{lem:coloring}.} 
\item {\bf The graph $\widetilde{H}$.} Let $V(\widetilde{H})=\{(R,l): R\in \{0,1,2,3\}^L, l\in L\}$,\footnote{That is, $R$ is a vector with $L$ entries where each entry is $0$, $1$, $2$ or $3$.} and $E(\widetilde{H})=\{\{(R,l),(R',l')\}: R[l'] \neq R'[l]\}$.
\item {\bf The labeling $\ell$.} For any $B\in V(\widetilde{G})$, let $\ell(B)$ contain all vertices $(R,l)\in V(\widetilde{H})$ such that $\widetilde{c}(B)=l$, and there exists $f: B\rightarrow \{1,2,3\}$ such that for all $i\in\{1,2,\ldots,L\}$, either $\phi_B(i)=R[i]=0$ or both $\phi_B(i)\neq 0$ and $f(\phi_B(i))=R[i]$. 
\end{itemize}
\end{definition}

Intuitively, for every vertex $B\in V(\widetilde{G})$, the function $\phi_B$ can be interpreted as follows. It is the assignment, for every possible color $i\in\{1,\ldots,L\}$, of the unique vertex $u$ within the vertex set identified with $B$ itself that is adjacent to some vertex in the vertex subset identified with some vertex $B'\in V(\widetilde{G})$ colored $i$, if such a vertex $u$ exists (else the assignment is of $0$). In a sense, $B$ thus stores the information on the identity of each vertex within it that is adjacent (in $G$) to some vertex outside of it, where each such internal vertex is uniquely accessed by specifying the color of the vertex in $\widetilde{G}$ whose identified vertex set contains the {\em neighbor}. With respect to the graph $\widetilde{H}$ and labeling $\ell$, we interpret each vertex $(R,l)\in V(\widetilde{H})$ as a ``placeholder'' (i.e.~potential assignment of the sought function $\varphi$) for any vertex $B\in V(\widetilde{G})$ that ``complies with the pattern encoded by the pair $(R,l)$'' as follows.  First and straightforwardly, $B$ must be colored $l$. Here, we remind that the colors of vertices in $\widetilde{G}$ belong to $\{1,\ldots,L\}$, while vertices in $G$ are colored $1$, $2$ or $3$ only. Then, the second requirement is that we can recolor (by $f$) the vertices in $B$ so that the color of each vertex in $B$ that is adjacent (in $G$) to some vertex outside $B$ is as encoded by the vector $R$---that is, for each color $i\in\{1,\ldots,L\}$, if the vertex $\phi_B(i)$ is defined (i.e., $\phi_B(i)\neq 0$), then its color (which is $1,2$ or $3$) must be equal to the $i$-th entry of $R$. (Further intuition is given in Fig.~\ref{fig:reduction} in Appendix \ref{app:detHomToISO}.)

Now, we state the correctness of the reduction.

\begin{lemma}[Lemma 3.3 in \cite{DBLP:journals/jacm/CyganFGKMPS17}]\label{lem:reduceCorrecr}
For any instance $G$ of {\sc 3-Coloring} where $G$ is an $n$-vertex graph of degree $d$, and a positive integer $r=o(\sqrt{|V(G)|})$, the instance $\reduce(G)=(\widetilde{G},\widetilde{H},\ell)$ is computable in  time polynomial in the sizes of $G,\widetilde{G}$ and $\widetilde{H}$, and has the following~properties. 
\begin{itemize}
\item $G$ is a \yes -instance of {\sc 3-Coloring} if and only if $(\widetilde{G},\widetilde{H},\ell)$ is a \yes -instance of {\sc LSH}.
\item $|V(\widetilde{G})|\leq n/r$, and $|V(\widetilde{H})|\leq \gamma(d)^r$ where $\gamma$ is some computable function of $d$.
\end{itemize}
\end{lemma}

We next prove that we can add colorings to the instance $\reduce(G)=(\widetilde{G},\widetilde{H},\ell)$ of {\sc LSH} in order to cast it as an instance of {\sc Prop-Col LSH} while making a minor mandatory modification to the graph $\widetilde{H}$.

\begin{lemma}\label{lem:reduceModification}
Given an instance $\reduce(G)=(\widetilde{G},\widetilde{H},\ell)$ of {\sc LSH}, an equivalent instance $(\widetilde{G},\widetilde{H}',$ $c_{\widetilde{G}},c_{\widetilde{H}'},\ell)$ of {\sc Prop-Col LSH}, where $\widetilde{H}'$ is a subgraph of $\widetilde{H}$, is computable in polynomial time.
\end{lemma}

\begin{proof}
Define $c_{\widetilde{G}} = \widetilde{c}$ where $\widetilde{c}$ is the coloring of $\widetilde{G}$ in Definition \ref{def:reduceLSH}. 
Additionally, let $\widetilde{H}'$ be the subgraph of $\widetilde{H}$ induced by the vertex set $\{(R,l)\in V(\widetilde{H}):$ there exists $B\in V(\widetilde{G})$ such that $(R,l)\in\ell(B)\}$. Then, define $c_{\widetilde{H}'}: V(\widetilde{H}') \rightarrow \{1,2,\ldots,L\}$ as follows: for any $(R,l)\in V(\widetilde{H}')$, define $c_{\widetilde{H}'}((R,l))=l$.  Notice that, by the definition of $V(\widetilde{H}')$, every set assigned by $\ell$ is subset of $V(\widetilde{H}')$.

First, we assert that $(\widetilde{G},\widetilde{H}',$ $c_{\widetilde{G}},c_{\widetilde{H}'},\ell)$ is an instance of {\sc Prop-Col LSH}. To this end, we need to verify that the three following properties hold.
\begin{enumerate}
\item $c_{\widetilde{G}}$ is a proper coloring of $\widetilde{G}$.
\item $c_{\widetilde{H}'}$ is a proper coloring of $\widetilde{H}'$.
\item For every $B\in V(\widetilde{G})$ and $(R,l)\in\ell(B)$, it holds that $c_{\widetilde{G}}(B)=c_{\widetilde{H}'}((R,l))$.
\end{enumerate}

By the definition of $c_{\widetilde{G}}$, it is a proper coloring of $\widetilde{G}^2$, which is a supergraph of $\widetilde{G}$. Thus, $c_{\widetilde{G}}$ is a proper coloring of $\widetilde{G}$.

Now, we argue that $c_{\widetilde{H}'}$ is a proper coloring of $\widetilde{H}'$. To this end, consider some edge $\{(R,l),(R',l')\}\in E(\widetilde{H}')$. We need to show that $c_{\widetilde{H}'}((R,l))\neq c_{\widetilde{H}'}((R',l'))$. By the definition of $c_{\widetilde{H}'}$, we have that $c_{\widetilde{H}'}((R,l))=l$ and $c_{\widetilde{H}'}((R',l'))=l'$, and therefore it suffices to show that $l\neq l'$. By the definition of $E(\widetilde{H})$ (which is a superset of $E(\widetilde{H}')$), we have that $R[l'] \neq R'[l]$. Thus, necessarily at least one among $R[l']$ and $R'[l]$ is not $0$, and so we suppose w.l.o.g.~that $R[l']$ is not $0$. Furthermore, since $(R,l)\in V(\widetilde{H}')$, we have that there exists $B\in E(\widetilde{G})$ such that $(R,l)\in\ell(B)$. Thus, 
\begin{itemize}
\item $\widetilde{c}(B)=l$.
\item There exists $f: B\rightarrow \{1,2,3\}$ such that for all $i\in\{1,2,\ldots,L\}$, either $\phi_B(i)=R[i]=0$ or both $\phi_B(i)\neq 0$ and $f(\phi_B(i))=R[i]$.
\end{itemize}
From the second property, and because $R[l']\neq 0$, we necessarily have that both $\phi_B(l')\neq 0$ and $f(\phi_B(l'))=R[l']$. In particular, by the definition of $\phi_B$, having $\phi_B(l')\neq 0$ means that there exists $(u,v,B')$ such that $u\in B$, $v\in B'$, $\{u,v\}\in E(G)$ and $\widetilde{c}(B')=l'$. By the definition of $\widetilde{G}$ as a grouping of $G$, having $u\in B$, $v\in B'$ and $\{u,v\}\in E(G)$ implies that $\{B,B'\}\in E(\widetilde{G})$. Because $\widetilde{c}$ is a proper coloring of $\widetilde{G}$, this means that $\widetilde{c}(B)\neq\widetilde{c}(B')$.  Since $\widetilde{c}(B)=l$ and $\widetilde{c}(B')=l'$, we derive that $l\neq l'$. Hence, $c_{\widetilde{H}'}$ is indeed a proper coloring of $\widetilde{H}'$.

To conclude that $(\widetilde{G},\widetilde{H}',c_{\widetilde{G}},c_{\widetilde{H}'},\ell)$ is indeed an instance of {\sc Prop-Col LSH}, it remains to assert that for every $B\in V(\widetilde{G})$ and $(R,l)\in\ell(B)$, it holds that $c_{\widetilde{G}}(B)=c_{\widetilde{H}'}((R,l))$. To this end, consider some $B\in V(\widetilde{G})$ and $(R,l)\in\ell(B)$. By the definition of $\ell$ (recall Definition \ref{def:reduceLSH}), $(R,l)\in\ell(B)$ implies that $\widetilde{c}(B)=l$. As $c_{\widetilde{G}}=\widetilde{c}$, we have that $c_{\widetilde{G}}(B)=l$. Moreover, the definition of  $c_{\widetilde{H}'}$ directly implies that $c_{\widetilde{H}'}((R,l))=l$. Thus, $c_{\widetilde{G}}(B)=c_{\widetilde{H}'}((R,l))$.

Finally, we argue that $(\widetilde{G},\widetilde{H},\ell)$ is a \yes -instance of {\sc LSH} if and only if $(\widetilde{G},\widetilde{H}',c_{\widetilde{G}},c_{\widetilde{H}'},\ell)$ is a \yes -instance of {\sc Prop-Col LSH}. In one direction, because $\widetilde{H}'$ is a subgraph of $\widetilde{H}$, it is immediate that if $(\widetilde{G},\widetilde{H}',$ $c_{\widetilde{G}},c_{\widetilde{H}'},\ell)$ is a \yes -instance of {\sc Prop-Col LSH}, then so is $(\widetilde{G},\widetilde{H},\ell)$. For the other direction, suppose that $(\widetilde{G},\widetilde{H},\ell)$ is a \yes -instance of {\sc LSH}. Thus, there exists a function $\varphi: V(\widetilde{G})\rightarrow V(\widetilde{H})$ such that {\em (i)} for every $\{B,B'\}\in E(\widetilde{G})$, $\{\varphi(B),\varphi(B')\}\in E(\widetilde{H})$, and {\em (ii)} for every $B\in V(\widetilde{G})$, $\varphi(B)\in \ell(B)$. In particular, directly by the definition of $V(\widetilde{H}')$, the second condition implies that for every $B\in V(\widetilde{G})$, it holds that $\varphi(B)\in V(\widetilde{H}')$. Thus, because $\widetilde{H}'$ is an induced subgraph of $\widetilde{H}$, it holds that for every $\{B,B'\}\in E(\widetilde{G})$, $\{\varphi(B),\varphi(B')\}\in E(\widetilde{H}')$. Therefore, $\varphi$ witnesses that $(\widetilde{G},\widetilde{H}',c_{\widetilde{G}},c_{\widetilde{H}'},\ell)$ is a \yes -instance of {\sc Prop-Col LSH}.
\end{proof}

We are now ready to assert the hardness of {\sc Prop-Col LSH}. The proof, based on Lemmas~\ref{lem:coloring}, \ref{lem:reduceCorrecr} and \ref{lem:reduceModification}, can be found in Appendix \ref{app:detHomToISO}.

\begin{lemma}\label{res:PropColLSIH}
Unless the ETH is false, there does not exist an algorithm that solves {\sc Prop-Col LSH} in time $n^{o(n)}$ where $n = \max(|V(G)|,|V(H)|)$.
\end{lemma}

\subparagraph{From Graph Homomorphism to Subgraph Isomorphism.}
In this part, we observe that the reduction of \cite{DBLP:journals/jacm/CyganFGKMPS17} from {\sc LSH} to {\sc LSI} can be essentially used as is to serve as a reduction  from {\sc Prop-Col LSH} to {\sc Prop-Col LSI}. For the sake of completeness, we give the full details (and the conclusion of the proof of Lemma~\ref{res:PropColLSI}) in Appendix~\ref{app:detHomToISO}.

\section{Lower Bound for the {\sc Cross Matching} Problem}\label{sec:crossMatching}

In this section, towards the proof of a lower bound for {\sc Clique Contraction}, we prove a lower bound for an intermediate problem called {\sc Cross Matching} that somewhat resembles {\sc Clique Contraction}, and which is defined as follows.

\defproblem{{\sc Cross Matching}}{A graph $G$ with a partition $(A,B)$ of $V(G)$ where $|A|=|B|$.}{Does there exist a perfect matching $M$ in $G$ such that every edge in $M$ has one endpoint in $A$ and the other in $B$, and $G/M$ is a clique?}

Our objective is to prove the following statement.

\begin{lemma}\label{res:CrossMatching}
Unless the ETH is false, there does not exist an algorithm that solves {\sc Cross Matching} in time $n^{o(n)}$ where $n=|A|$.
\end{lemma}

\begin{proof}
Towards a contradiction, suppose that there exists an algorithm, denoted by {\sf MatchingAlg}, that solves {\sc Cross Matching} in time $n^{o(n)}$ where $n$ is the number of vertices in the set $A$ in the input. We will show that this implies the existence of an algorithm, denoted by {\sf LSIAlg}, that solves {\sc Prop-Col LSI} in time $n^{o(n)}$ where $n$ is the number of vertices in the input graph $G$, thereby contradicting Lemma~\ref{res:PropColLSI} and hence completing the proof.

\begin{figure}[t]
 \center{\includegraphics[scale=0.75]{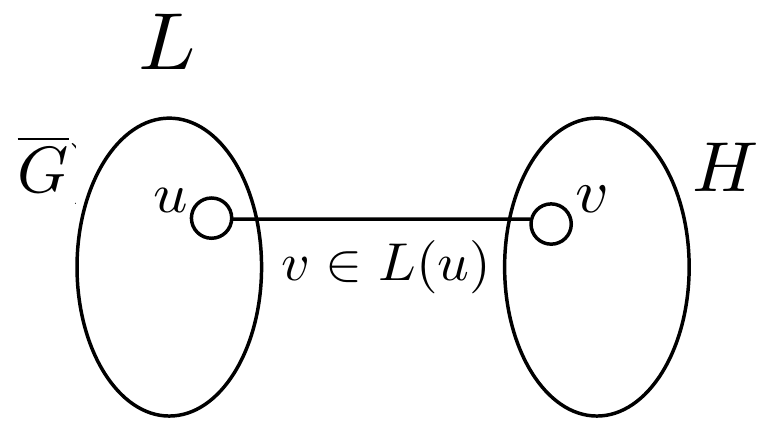}}
\caption{\label{fig:crossMatching} The construction of an instance of {\sc Cross Matching} in the proof of Lemma~\ref{res:CrossMatching}.}
\end{figure}

We define the execution of {\sf LSIAlg} as follows. Given an instance $(G,H,c_G,c_H,\ell)$ of {\sc Prop-Col LSI}, {\sf LSIAlg} constructs an instance $(L,A,B)$ of {\sc Cross Matching} as follows (see Fig.~\ref{fig:crossMatching}):
\begin{itemize}
\item $V(L)=V(\overline{G})\cup V(H)$.
\item $E(L)=E(\overline{G})\cup E(H)\cup \{\{u,v\}: u\in V(G), v\in L(u)\}$.
\item $A=V(\overline{G})$ and $B=V(H)$.
\end{itemize}
Then, {\sf LSIAlg} calls {\sf MatchingAlg} with $(L,A,B)$ as input, and returns the answer of this call.

Denote $n=|V(G)|$, and notice that $|A|=|B|=n$. Thus, because {\sf MatchingAlg} runs in time $|A|^{o(|A|)}=n^{o(n)}$, so does {\sf LSIAlg}.

For the correctness of the algorithm, first suppose that $(G,H,c_G,c_H,\ell)$ is a \yes -instance of {\sc Prop-Col LSI}. This means that there exists a bijective function $\varphi: V(G)\rightarrow V(H)$ such that {\em (i)} for every $\{u,v\}\in E(G)$, $\{\varphi(u),\varphi(v)\}\in E(H)$, and {\em (ii)} for every $u\in V(G)$, $\varphi(u)\in L(u)$.  Having $\varphi$ at hand, we will show that $(L,A,B)$ is a \yes -instance, which will imply that the call to {\sf MatchingAlg} with $(L,A,B)$ as input returns \yes, and hence {\sf LSIAlg} returns \yes.

Based on $\varphi$, we define a subset $M\subseteq E(L)$ as follows: $M=\{\{u,\varphi(u)\}: u\in A\}$. Notice that the containment of $M$ in $E(L)$ follows from the definition of $E(L)$ and Condition {\em (ii)} above. Moreover, by the definition of $A$, $B$ and because $\varphi$ is bijective, it further follows that $M$ is a perfect matching in $L$ such that every edge in $M$ has one endpoint in $A$ and the other in $B$. Thus, to conclude that $(L,A,B)$ is a \yes-instance, it remains to argue that $L/M$ is a clique. To this end, we consider two arbitrary vertices $x$ and $y$ of $L/M$, and prove that they are adjacent in $L/M$. Necessarily $x$ is a vertex that replaced two vertices $u\in A$ and $u'\in B$ such that $\{u,u'\}\in M$, and $y$ is a vertex that replaced two vertices $v\in A\setminus\{u\}$ and $v'\in B\setminus\{u'\}$ such that $\{v,v'\}\in M$. By the definition of contraction, to show that $x$ and $y$ are adjacent in $L/M$, it suffices to show that $u$ and $v$ are adjacent in $L$ or $u'$ and $v'$ are adjacent in $L$ (or both). To this end, suppose that $u$ and $v$ are not adjacent in $L$, else we are done. By the definition of $E(L)$, this means that $\{u,v\}\notin E(\overline{G})$ and hence $\{u,v\}\in E(G)$. By Condition {\em (i)} above, we derive that $\{\varphi(u),\varphi(v)\}\in E(H)$. By the definition of $M$, we know that $u'=\varphi(u)$ and $v'=\varphi(v)$, therefore  $\{u',v'\}\in E(H)$. In turn, by the definition of $E(L)$, we get that $\{u',v'\}\in E(L)$. Thus, the proof of the forward direction is complete.

Now, suppose that {\sf LSIAlg} returns \yes, which means that the call to {\sf MatchingAlg} with $(L,A,B)$ returns \yes.  Thus, $(L,A,B)$ is a \yes -instance, which means that there exists a perfect matching $M$ in $G$ such that every edge in $M$ has one endpoint in $A$ and the other in $B$, and $G/M$ is a clique. We define a function $\varphi: A\rightarrow B$ as follows. For every $u\in V(G)$, let $\varphi(u)=v$ where $v$ is the unique vertex in $B$ such that $\{u,v\}\in M$; the existence and uniqueness of $v$ follows from the supposition that $M$ is a perfect matching such that every edge in $M$ has one endpoint in $A$ and the other in $B$. Furthermore, by the definition of $A,B$ and the edges in $E(L)$ with one endpoint in $A$ and the other in $B$, it directly follows that $\varphi$ is a bijective mapping between $V(G)$ and $V(H)$ such that for every $u\in V(G)$, it holds that $\varphi(u)\in L(u)$. Thus, it remains to argue that for every edge $\{u,v\}\in E(G)$, it holds that $\{\varphi(u),\varphi(v)\}\in E(H)$. To this end, consider some arbitrary edge $\{u,v\}\in E(G)$, and denote $u'=\varphi(u)$ and $v'=\varphi(v)$. Because $L/M$ is a clique and $M$ is a matching that, by the definition of $\varphi$, necessarily contains both $\{u,u'\}$ and $\{v,v'\}$, we derive that at least one of the following four conditions must be satisfied: {\em (i)} $\{u,v\}\in E(L)$; {\em (ii)} $\{u',v'\}\in E(L)$; {\em (iii)} $\{u,v'\}\in E(L)$; {\em (iv)} $\{v,u'\}\in E(L)$.  Because $\{u,v\}\in E(G)$, we have that $\{u,v\}\notin E(\overline{G})$ and therefore $\{u,v\}\notin E(L)$. Thus, we are left with Conditions {\em (ii)}, {\em (iii)} and {\em (iv)}. Now, we will crucially rely on the proper colorings of $G$ and $H$ to rule out the satisfaction of Conditions {\em (iii)} and {\em (iv)}. 

\begin{claim}\label{claim:CrossMatching}
For any two edges $\{x,x'\},\{y,y'\}\in E(L)$ such that $\{x,y\}\in E(G)$ and $x',y'\in V(H)$, it holds that neither $\{x,y'\}$ nor $\{y,x'\}$ belongs to $E(L)$.
\end{claim}

\noindent{\em Proof of Claim \ref{claim:CrossMatching}.} Because $c_G$ is a proper coloring of $G$ and $\{x,y\}\in E(G)$, it holds that $c_G(x)\neq x_G(y)$. Because $\{x,x'\},\{y,y'\}\in E(L)$, $x,y\in V(G)$ and $x',y'\in V(H)$, and by the definition of $E(L)$, it holds that $x'\in L(x)$ and $y'\in L(y)$, and therefore $c_G(x)=c_H(x')$ and $c_G(y)=c_H(y')$. Thus, $c_G(x)\neq c_H(y')$ and $c_G(y)\neq c_H(x')$, implying that $y'\notin L(x)$ and $x'\notin L(y)$. In turn, by the definition of $E(L)$, this means that neither $\{x,y'\}$ nor $\{y,x'\}$ belongs to $E(L)$. This completes the proof of the claim.
$\diamond$

\medskip
\noindent We now return to the proof of the lemma. By Claim \ref{claim:CrossMatching}, we are only left with Condition {\em (ii)}, that is, $\{u',v'\}\in E(L)$. However, by the definition of $E(L)$, this means that $\{u',v'\}\in E(H)$. As argued earlier, this completes the proof of the reverse direction.
\end{proof}


\section{Lower Bounds: {\sc Clique Contraction} and {\sc Hadwiger Number}}\label{sec:hadwiger}

In this section, we prove a lower bound for {\sc Clique Contraction} and consequently for {\sc Hadwiger Number}, defined as follows.

\defproblem{{\sc Clique Contraction}}{A graph $G$ and $t\in\mathbb{N}$.}{Is there a subset $F\subseteq E(G)$ of size at most $t$ such that $G/F$ is a clique?}

\defproblem{{\sc Hadwiger Number}}{A graph $G$ and $h\in\mathbb{N}$.}{Is the Hadwiger number of $G$ at least as large as $h$?}

Our objective is to prove the following statement, where the analogous statement for {\sc Hadwiger Number} (called Theorem \ref{thm:informalmain} in the introduction) will follow as a corollary.

\begin{theorem}\label{res:CliqueContraction}
Unless the ETH is false, there does not exist an algorithm that solves {\sc Clique Contraction} in time $n^{o(n)}$ where $n=|V(G)|$.
\end{theorem}

To make our approach adaptable to extract analogous statements for other contraction problems, we will first define a new problem called {\sc Noisy Structured Clique Contraction} (which will arise in Appendix \ref{sec:contractionToClasses}) along with a special case of it that is also a special case of {\sc Clique Contraction}. Then, we will prove a crucial property of instances of {\sc Noisy Structured Clique Contraction}, and afterwards we will use this property to prove Theorem \ref{res:CliqueContraction} and its corollary. The definition of the new problem is as follows (see Fig.~\ref{fig:noisyStructured}).

\begin{figure}[t]
 \center{\includegraphics[scale=0.6]{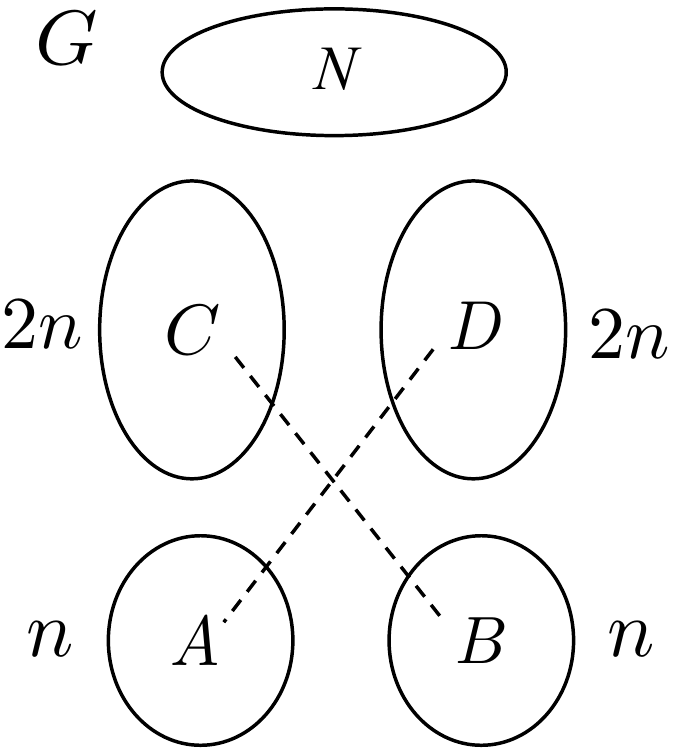}}
\caption{\label{fig:noisyStructured} An instance of {\sc Noisy Structured Clique Contraction} where dashed lines represent non-edges.}
\end{figure}

\defproblem{{\sc Noisy Structured Clique Contraction}}{A graph $G$ on at least $6n$ vertices for some $n\in\mathbb{N}$, and a partition $(A,B,C,D,N)$ of $V(G)$ such that $|A|=|B|=n$, $|C|=|D|=2n$, 
no vertex in $A$ is adjacent to any vertex in $D$, and no vertex in $B$ is adjacent to any vertex in $C$.}{Does there exist a subset $F\subseteq E(G)$ of size at most $n$ such that $G[A\cup B\cup C\cup D\cup X]/F$ is a clique,\footnote{Note that $F$ might contain edges outside $G[A\cup B\cup C\cup D\cup X]$. Then, we slightly abuse notation so that $G[A\cup B\cup C\cup D\cup X]/F$ refers to $G[A\cup B\cup C\cup D\cup X]/(F\cap E(G[A\cup B\cup C\cup D\cup X]))$.} where $X=\{u\in N:$ there exists a vertex $v\in A\cup B\cup C\cup D$ such that $u$ and $v$ belong to the same connected component of $G[F]\}$?}

Intuitively, the vertex set $X$ consists of the noise (represented by $N$) that ``interacts'' with non-noise (represented by $V(G)\setminus N$) through contracted edges (in $F$), i.e.~the vertices in $N$ that lie together with at least one vertex in $V(G)\setminus N$ in a component that will be contracted and thereby replaced by a single vertex. We refer to the special case of {\sc Noisy Structured Clique Contraction} where $N=\emptyset$ as {\sc Structured Clique Contraction}. Note that {\sc Structured Clique Contraction} is also a special case of {\sc Clique Contraction}.

Solutions to instances of {\sc Noisy Structured Clique Contraction} exhibit the following property, which will be crucial in the proof of Theorem~\ref{res:CliqueContraction} as well as results in Section \ref{sec:contractionToClasses}.

\begin{lemma}\label{lem:propNoisyStructured}
Let $F$ be a solution to an instance $(G,A,B,C,D,N,n)$ of {\sc Noisy Structured Clique Contraction}. Then, $F$ is a matching of size $n$ in $G$ such that each edge in $F$ has one endpoint in $A$ and the other in $B$.
\end{lemma}

\begin{proof}
We first argue that every vertex in $A\cup B$ is incident to at least one edge in $F$. Targeting a contradiction, suppose that there exists a vertex $u\in A\cup B$ that is not incident to any edge in $F$.  Because $|A\cup B\cup C\cup D|=6n$, $|F|\leq n$ and $G[A\cup B\cup C\cup D\cup X]/F$ is a clique (where the last two properties follow from the supposition that $F$ is a solution), it holds that $G[A\cup B\cup C\cup D\cup X]/F$ is a clique on at least $5n+|X|$ vertices. Hence, the degree of every vertex in $G[A\cup B\cup C\cup D\cup X]/F$, and in particular of $u$, should be $5n-1+|X|$ in $G[A\cup B\cup C\cup D\cup X]/F$. However, because no vertex in $A$ is adjacent to any vertex in $D$ and no vertex in $B$ is adjacent to any vertex in $C$, the degree of any vertex in $A\cup B$, and in particular of $u$, is at most $|A\cup B|-1 + |C\cup D|/2 + |X|=4n-1+|X|$ in $G[A\cup B\cup C\cup D\cup X]$. Because $u$ is not incident to any edge in $F$, its degree in $G[A\cup B\cup C\cup D\cup X]/F$ is at most its degree in $G[A\cup B\cup C\cup D\cup X]$. This is a contradiction, thus we get that indeed every vertex in $A\cup B$ is incident to at least one edge in $F$. From this, because $|F|\leq n$ and $|A\cup B|=2n$, we derive that $F$ is a perfect matching in $G[A\cup B]$.

It remains to argue that every edge in $F$ has one endpoint in $A$ and the other in $B$. Targeting a contradiction, suppose that this is false. Because $F$ is a perfect matching in $G[A\cup B]$, this means that there exist two vertices $a,a'\in A$ such that $\{a,a'\}\in F$. By the definition of {\sc Noisy Structured Clique Contraction}, neither $a$ nor $a'$ is adjacent to any vertex in $D$. Moreover, note that $D\subseteq V(G[A\cup B\cup C\cup D\cup X]/F)$. In particular, the vertex of $G[A\cup B\cup C\cup D\cup X]/F$ yielded by the contraction of $\{a,a'\}$ is not adjacent to any vertex of $D$ in $G[A\cup B\cup C\cup D\cup X]/F$. However, this is a contradiction because $G[A\cup B\cup C\cup D\cup X]/F$ is a clique.
\end{proof}

We now prove a lower bound for {\sc Structured Clique Contraction}. Because it is a special case of {\sc Clique Contraction}, this will directly yield the correctness of Theorem~\ref{res:CliqueContraction}.

\begin{lemma}\label{res:StructCliqueContraction}
Unless the ETH is false, there does not exist an algorithm that solves {\sc Structured Clique Contraction} in time $n^{o(n)}$ where $n=|V(G)|$.
\end{lemma}

\begin{proof}
Targeting a contradiction, suppose that there exists an algorithm, denoted by {\sf CliConAlg}, that solves {\sc Structured Clique Contraction} in time $n^{o(n)}$ where $n$ is the number of vertices in the input graph. We will show that this implies the existence of an algorithm, denoted by {\sf MatchingAlg}, that solves {\sc Cross Matching} in time $n^{o(n)}$ where $n$ is the size of the set $A$ in the input, thereby contradicting Lemma~\ref{res:CrossMatching} and hence completing the proof.

\begin{figure}[t]
 \center{\includegraphics[scale=0.6]{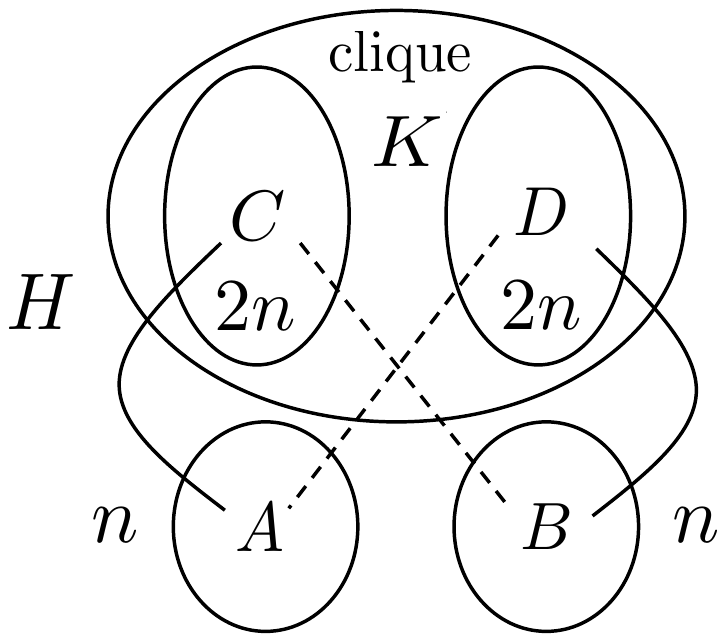}}
\caption{\label{fig:structuredCC} The construction of an instance of {\sc Structured Clique Contraction} in the proof of Lemma~\ref{res:StructCliqueContraction} where dashed lines represent non-edges.}
\end{figure}

We define the execution of {\sf MatchingAlg} as follows. Given an instance $(G,A,B)$ of {\sc Cross Matching}, {\sf MatchingAlg} constructs an instance $(H,A,B,C,D,n)$ of {\sc Structured Clique Contraction} as follows (see Fig.~\ref{fig:structuredCC}):
\begin{itemize}
\item Let $n=|A|$, and $K$ be a clique on $4n$ new vertices. Let $(C,D)$ be a partition of $V(K)$ such that $|C|=|D|$.
\item $V(H)=V(G)\cup V(K)$.
\item $E(H)=E(G)\cup E(K)\cup \{\{a,c\}: a\in A, c\in C\}\cup \{\{b,d\}: b\in B, d\in D\}$.
\end{itemize}
Then, {\sf MatchingAlg} calls {\sf CliConAlg} with $(H,A,B,C,D,n)$ as input, and returns the answer.

First, note that by construction, $|V(H)|=6n$. Thus, because {\sf CliConAlg} runs in time $|V(H)|^{o(|V(H)|)} \leq n^{o(n)}$, it follows that {\sf MatchingAlg} runs in time $n^{o(n)}$.

For the correctness of the algorithm, first suppose that $(G,A,B)$ is a \yes -instance of {\sc Cross Matching}. This means that there exists a perfect matching $M$ in $G$ such that every edge in $M$ has one endpoint in $A$ and the other in $B$, and $G/M$ is a clique. By the definition of $E(H)$, $M\subseteq E(H)$. We will show that  $H/M$ is a clique. As $|M|=n$, this will mean that $(H,A,B,C,D,n)$ is a \yes -instance of {\sc Structured Clique Contraction}, which will mean, in turn, that the call to {\sf CliConAlg} with $(H,A,B,C,D,n)$ as input returns \yes, and hence {\sf MatchingAlg} returns \yes.

Note that $V(H/M)=V(K)\cup V(G/M)$. To show that $H/M$ is a clique, we consider two arbitrary vertices $u,v\in V(H/M)$, and show that they are adjacent in $H/M$. If $u,v\in V(K)$, then because $K$ is a clique, it is clear that $\{u,v\}\in E(H/M)$. Moreover, if $u,v\in G/M$, then because $G/M$ is a clique, it is clear that $\{u,v\}\in E(H/M)$. Thus, one of the vertices $u$ and $v$ belongs to $V(G/M)$ and the other belongs to $V(K)$. We suppose w.l.o.g.~that  $u\notin V(K)$. Because $M$ is a perfect matching in $G$ such that every edge in $M$ has one endpoint in $A$ and the other in $B$, it follows that $u$ resulted from the contraction of the edge between some $a\in A$ and some $b\in B$. If $v\in C$, then $\{a,v\}\in E(H)$, and otherwise $v\in D$ and so $\{b,v\}\in E(H)$. Thus, by the definition of contraction, we conclude that $\{u,v\}\in E(H/M)$. This completes the proof of the forward direction.

Now, suppose that {\sf MatchingAlg} returns \yes, which means that the call to {\sf CliConAlg} with $(H,A,B,C,D,n)$ returns \yes.  Thus, $(H,A,B,C,D,n)$ is a \yes -instance, which means that there exists a subset $F\subseteq E(H)$ of size at most $n$ such that $H/F$ is a clique. We will show that $F$ is a perfect matching in $G$ such that every edge in $F$ has one endpoint in $A$ and the other in $B$. Because $H/F$ is a clique, this will imply that $G/F$ is a clique and thus that $(G,A,B)$ is a \yes -instance of {\sc Cross Matching}. To achieve this, notice that by Lemma \ref{lem:propNoisyStructured}, $F$ is a matching of size $n$ in $H$ such that each edge in $F$ has one endpoint in $A$ and the other in $B$. Because $G=H[A\cup B]$, we have that $F$ is a perfect matching in $G$. Thus, the proof of the reverse direction is complete.
\end{proof}

\begin{corollary}\label{res:HadwigerNumber}
Unless the ETH is false, there does not exist an algorithm that solves {\sc Hadwiger Number} in time $n^{o(n)}$ where $n=|V(G)|$.
\end{corollary}

\begin{proof}
Targeting a contradiction, suppose that there exists an algorithm, denoted by {\sf HadwigerAlg}, that solves {\sc Hadwiger Number} in time $n^{o(n)}$ where $n$ is the number of vertices in the input graph. We will show that this implies the existence of an algorithm, denoted by {\sf CliConAlg}, that solves {\sc Clique Contraction} in time $n^{o(n)}$ where $n$ is the number of vertices in the input graph, thereby contradicting Theorem \ref{res:CliqueContraction} and hence completing the proof.

We define the execution of {\sf CliConAlg} as follows. Given an instance $(G,t)$ of {\sc Clique Contraction}, if $G$ is not connected, then {\sf CliConAlg}  returns \no, and otherwise it returns \yes\ if and only if {\sf HadwigerAlg} returns \yes\ when called with $(G,|V(G)|-t)$ as input. Because the call to {\sf HadwigerAlg} with input $(G,|V(G)|-t)$ runs in time $n^{o(n)}$ where $n=|V(G)|$, we have that {\sf CliConAlg} runs in time $n^{o(n)}$ as well.

For the correctness of the algorithm, first observe that if $G$ is not connected, then no sequence of edge contractions can yield a clique, and hence it is correct to return \no. Thus, now assume that $G$ is connected. First, suppose that $(G,t)$ is a \yes -instance of {\sc Clique Contraction}. This means that there exists a sequence of at most $t$ edge contractions that transforms $G$ into a clique. In particular, this clique must have at least $|V(G)|-t$ vertices, and therefore the Hadwiger number of $G$ is at least as large as $|V(G)|-t$. By the correctness of {\sf HadwigerAlg}, its call with $(G,|V(G)|-t)$ returns \yes, and therefore {\sf CliConAlg} returns \yes.

Now, suppose that {\sf CliConAlg} returns \yes, which means that the call to {\sf HadwigerAlg} with $(G,|V(G)|-t)$ returns \yes. By the correctness of {\sf HadwigerAlg}, the clique $K_h$ for $h=|V(G)|-t$ is a minor of $G$. This means that there is a sequence of vertex deletions, edge deletions and edge contractions that transforms $G$ into $K_h$. In particular, this sequence can contain at most $t$ vertex deletions and edge contractions in total. Furthermore, by replacing each vertex deletion for a vertex $v$ by an edge contraction for some edge $e$ incident to $v$ (which exists because $G$ is connected) and dropping all edge deletions, we obtain another sequence that transforms $G$ into $K_h$. Because this sequence contains only edge contractions, and at most $t$ of them, we conclude that $(G,t)$ is a \yes -instance of {\sc Clique Contraction}.
\end{proof}

\bibliography{references,subiso,book_pc,hom}

\begin{thebibliography}{10}

\bibitem{agrawal2017split}
Akanksha Agrawal, Daniel Lokshtanov, Saket Saurabh, and Meirav Zehavi.
\newblock Split contraction: {T}he untold story.
\newblock In {\em 34th Symposium on Theoretical Aspects of Computer Science
  (STACS 2017)}. Schloss Dagstuhl-Leibniz-Zentrum fuer Informatik, 2017.

\bibitem{Bjorklund10}
Andreas Bj{\"{o}}rklund.
\newblock Determinant sums for undirected hamiltonicity.
\newblock {\em {SIAM} J. Comput.}, 43(1):280--299, 2014.

\bibitem{BjorklundHK2009-Se}
Andreas Bj{\"o}rklund, Thore Husfeldt, and Mikko Koivisto.
\newblock Set partitioning via inclusion--exclusion.
\newblock {\em SIAM J. Computing}, 39(2):546--563, 2009.

\bibitem{MR593989}
B.~Bollob\'{a}s, P.~A. Catlin, and P.~Erd\H{o}s.
\newblock Hadwiger's conjecture is true for almost every graph.
\newblock {\em European J. Combin.}, 1(3):195--199, 1980.
\newblock URL: \url{https://doi.org/10.1016/S0195-6698(80)80001-1}, \href
  {http://dx.doi.org/10.1016/S0195-6698(80)80001-1}
  {\path{doi:10.1016/S0195-6698(80)80001-1}}.

\bibitem{Chen2005216}
Jianer Chen, Benny Chor, Michael~R. Fellows, Xiuzhen Huang, David Juedes,
  Iyad~A. Kanj, and Ge~Xia.
\newblock Tight lower bounds for certain parameterized {NP}-hard problems.
\newblock {\em Information and Computation}, 201(2):216 -- 231, 2005.
\newblock URL:
  \url{http://www.sciencedirect.com/science/article/pii/S0890540105000763},
  \href {http://dx.doi.org/http://dx.doi.org/10.1016/j.ic.2005.05.001}
  {\path{doi:http://dx.doi.org/10.1016/j.ic.2005.05.001}}.

\bibitem{DBLP:journals/jacm/CyganFGKMPS17}
Marek Cygan, Fedor~V. Fomin, Alexander Golovnev, Alexander~S. Kulikov, Ivan
  Mihajlin, Jakub Pachocki, and Arkadiusz Socala.
\newblock Tight lower bounds on graph embedding problems.
\newblock {\em J. {ACM}}, 64(3):18:1--18:22, 2017.
\newblock URL: \url{https://doi.org/10.1145/3051094}, \href
  {http://dx.doi.org/10.1145/3051094} {\path{doi:10.1145/3051094}}.

\bibitem{DowneyF99}
Rodney~G. Downey and Michael~R. Fellows.
\newblock {\em Parameterized complexity}.
\newblock Springer-Verlag, New York, 1999.

\bibitem{Fomin:2010mo}
Fedor~V. Fomin and Dieter Kratsch.
\newblock {\em Exact Exponential Algorithms}.
\newblock Springer, 2010.
\newblock An EATCS Series: Texts in Theoretical Computer Science.

\bibitem{Golumbic:2004:AGT:984029}
Martin~Charles Golumbic.
\newblock {\em Algorithmic Graph Theory and Perfect Graphs}.
\newblock North-Holland Publishing Co., Amsterdam, The Netherlands, The
  Netherlands, 2004.

\bibitem{DBLP:conf/stoc/GroheKMW11}
Martin Grohe, Ken{-}ichi Kawarabayashi, D{\'{a}}niel Marx, and Paul Wollan.
\newblock Finding topological subgraphs is fixed-parameter tractable.
\newblock In {\em Proceedings of the 43rd {ACM} Symposium on Theory of
  Computing, {STOC} 2011, San Jose, CA, USA, 6-8 June 2011}, pages 479--488,
  2011.

\bibitem{DBLP:conf/coco/ImpagliazzoP99}
Russell Impagliazzo and Ramamohan Paturi.
\newblock Complexity of k-sat.
\newblock In {\em Proceedings of the 14th Annual {IEEE} Conference on
  Computational Complexity, Atlanta, Georgia, USA, May 4-6, 1999}, pages
  237--240, 1999.

\bibitem{DBLP:journals/jcss/ImpagliazzoPZ01}
Russell Impagliazzo, Ramamohan Paturi, and Francis Zane.
\newblock Which problems have strongly exponential complexity?
\newblock {\em J. Comput. Syst. Sci.}, 63(4):512--530, 2001.

\bibitem{Lawler76}
Eugene~L. Lawler.
\newblock A note on the complexity of the chromatic number problem.
\newblock {\em Information Processing Letters}, 5(3):66--67, 1976.

\bibitem{LingasW09}
Andrzej Lingas and Martin Wahlen.
\newblock An exact algorithm for subgraph homeomorphism.
\newblock {\em J. Discrete Algorithms}, 7(4):464--468, 2009.
\newblock URL: \url{http://dx.doi.org/10.1016/j.jda.2008.10.003}, \href
  {http://dx.doi.org/10.1016/j.jda.2008.10.003}
  {\path{doi:10.1016/j.jda.2008.10.003}}.

\bibitem{RobertsonS-GMXIII}
Neil Robertson and Paul~D. Seymour.
\newblock Graph minors. {XIII}. {T}he disjoint paths problem.
\newblock {\em J. Combinatorial Theory Ser. B}, 63(1):65--110, 1995.

\bibitem{T2008}
Patrick Traxler.
\newblock The time complexity of constraint satisfaction.
\newblock In {\em Parameterized and Exact Computation}, pages 190--201.
  Springer, 2008.

\end{thebibliography}

\appendix


\section{Preliminaries}\label{sec:prelims}

For a vector $R$ with $L$ entries and $i\in\{1,\ldots,L\}$, let $R[i]$ be the value of the $i$-th entry of $R$. Unless specified otherwise, bases of logarithms are assumed to be $2$. 

Given a graph $G$, let $V(G)$ and $E(G)$ denote its vertex set and edge set, respectively. Given a subset $U\subseteq V(G)$, let $G[U]$ denote the subgraph of $G$ induced by $U$, that is, $V(G[U])=U$ and $E(G[U])=\{\{u,v\}\in E(G): u,v\in U\}$. Given a subset $F\subseteq E(G)$, let $V(F)$ denote the set of vertices that are incident in $G$ to at least one edge in $F$, and let $G[F]=G[V(F)]$. We say that $G$ contains a graph $H$ as an induced subgraph if there exists $U\subseteq V(G)$ such that $G[U]$ and $H$ are identical up to relabelling vertices (more precisely, isomorphic). The set of neighbors of a vertex $u\in V(G)$ is denoted by $N_G(u)$, that is, $N_G(u)=\{v\in V(G): \{u,v\}\in E(G)\}$. When $G$ is clear from context, we drop it from subscripts of notations. 
A {\em matching} $M$ in $G$ is subset of $E(G)$ such that no two edges in $M$ share an endpoint. In case every vertex in $V(G)$ is an endpoint of an edge in $M$, that is, $|M|=|V(G)|/2$, it is said that $M$ is {\em perfect}. A function $c: V(G)\rightarrow \mathbb{N}$ is a {\em proper coloring} of $G$ if for every edge $\{u,v\}\in E(G)$, $c(u)\neq c(v)$. The {\em complement} of $G$, denoted by $\overline{G}$, is the graph with vertex set $V(G)$ and edge set $\{\{u,v\}\notin E(G): u,v\in V(G), u\neq v\}$.

Given an edge $e=\{u,v\}\in E(G)$, the {\em contraction} of $e$ in $G$ is the operation that replaces $u$ and $v$ by a new vertex that is adjacent to all vertices previously adjacent to $u$ or $v$ (or both), where the resulting graph is denoted by $G/e$. In other words, $V(G/e)=(V(G)\setminus\{u,v\})\cup\{x\}$ for some new vertex $x$, and $E(G/e)=\{\{s,t\}\in E(G): s,t\notin\{u,v\}\}\cup \{\{s,x\}: s\in N(u)\cup N(v)\}$. More generally, given a subset $F\subseteq E(G)$, the {\em contraction} of $F$ in $G$ is the operation that replaces each connected component $C$ of $G[F]$ by a new vertex $x_C$ that is adjacent to all vertices previously adjacent to at least one vertex in $C$, where the resulting graph is denoted by $G/F$.
A graph $H$ is said to be a {\em minor} of a graph $G$ if $H$ can be obtained from $G$ by a series of vertex deletions, edge deletions and edge contractions. For any $h\in\mathbb{N}$, the clique on $h$ vertices is denoted by $K_h$, and the cycle on $h$ vertices is denoted by $C_h$. The {\em Hadwiger number} of a graph $G$ is the largest $h\in\mathbb{N}$ such that $K_h$ is a minor of $G$.

To obtain (essentially) tight conditional lower bounds for the running times of algorithms, we rely on the {\em Exponential-Time Hypothesis (ETH)}~\cite{DBLP:conf/coco/ImpagliazzoP99,DBLP:journals/jcss/ImpagliazzoPZ01}. To formalize its statement, we remind that given a formula $\varphi$ in conjuctive normal form (CNF) with $n$ variables and $m$ clauses, the task of {\sc CNF-SAT} is to decide whether there is a truth assignment to the variables that satisfies $\varphi$. In the {\sc $p$-CNF-SAT} problem, each clause is restricted to have at most $p$ literals. Then, ETH asserts that {\sc 3-CNF-SAT} cannot be solved in time $2^{o(n)}$. 

\section{Details Omitted from Section \ref{sec:listSI}}\label{app:detHomToISO}

\begin{figure}[t]
 \center{\includegraphics[scale=0.6]{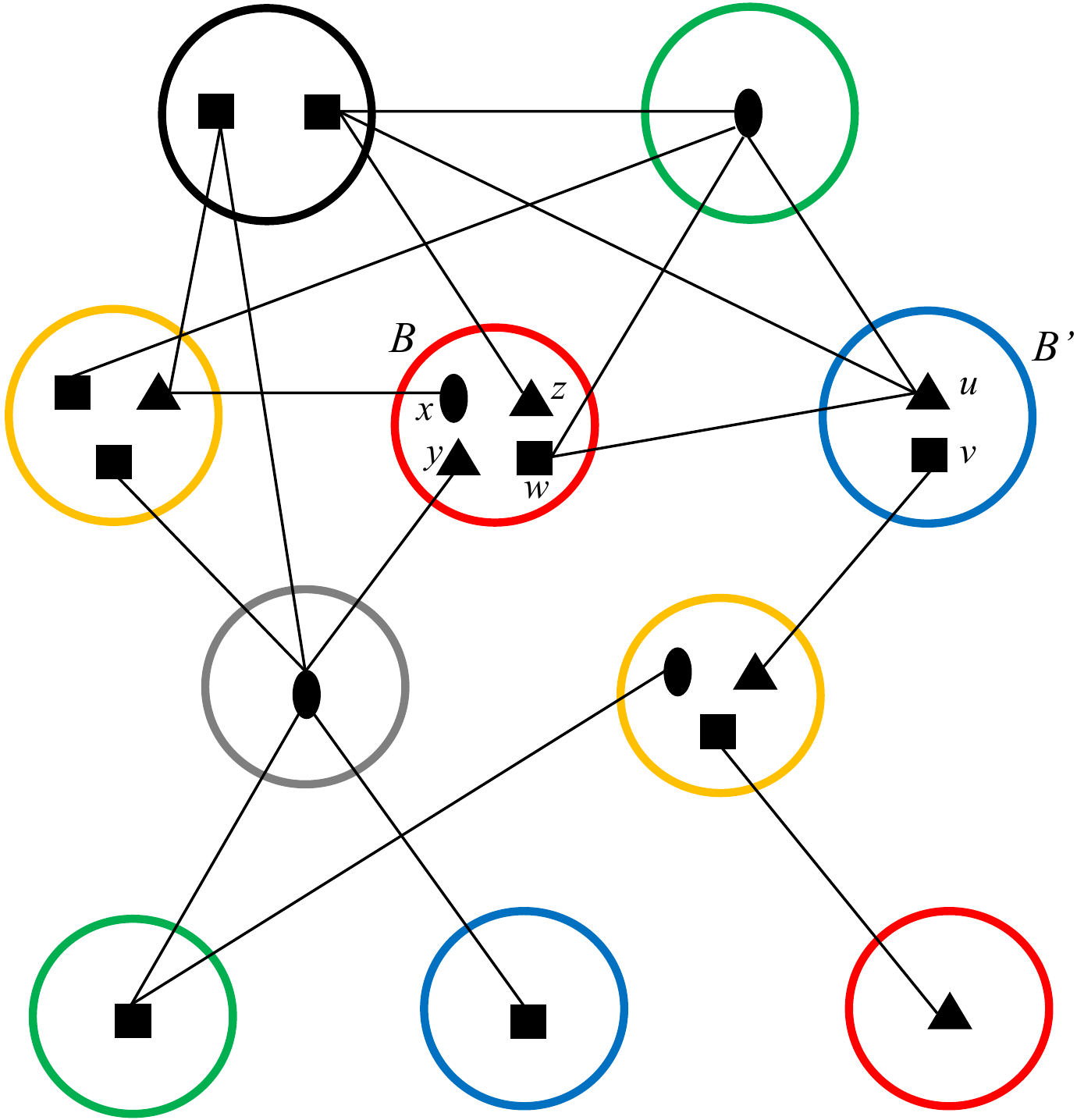}}
\caption{\label{fig:reduction} The reduction in Definition \ref{def:reduceLSH}. The vertices  of $G$ are depicted by black shapes, where each distinct shape represents a different color (say, square is $1$, rectangle is $2$ and oval is $3$), and the vertices of $\widetilde{G}$ are depicted by circles enclosing the vertex sets identifies with them, where the color of a vertex is the color of its circle (say, black is $1$, green is $2$, yellow is $3$, red is $4$, blue is $5$ and grey is $6$). Edges (of both graphs) are depicted by black lines. (The graph $\widetilde{H}$ is not shown). Then, the function $\phi_B$ is defined as follows: $\phi_B(1)=z, \phi_B(2)=\phi_B(5)=w, \phi_B(3)=x,\phi_B(4)=0,$ and $\phi_B(6)=y$. Moreover, the function $\phi_{B'}$ is defined as follows: $\phi_{B'}(1)=\phi_{B'}(2)=\phi_{B'}(4)=u, \phi_{B'}(3)=v,$ and $\phi_{B'}(5)=\phi_{B'}(6)=0$. With respect to $B$ and $B'$, the labeling $\ell$ is defined as follows: $\ell(B)=\{(R,4): R[1]\neq 0,R[2]=R[5]\neq 0, R[3]\neq 0, R[4]=0,R[6]\neq 0\}$, and $\ell(B')=\{(R,5): R[1]=R[2]=R[4]\neq 0, R[3]\neq 0, R[5]=R[6]=0\}$.}
\end{figure}

We first present the proof of Lemma \ref{res:PropColLSIH}.

\begin{proof}[Proof of Lemma \ref{res:PropColLSIH}]
Targeting a contradiction, suppose that there exists an algorithm, denoted by {\sf LSHAlg}, that solves {\sc Prop-Col LSH} in time $n^{o(n)}$ where $n = \max(|V(G)|,|V(H)|)$ for input graphs $G$ and $H$. We will show that this implies the existence of an algorithm, denoted by {\sf ColAlg}, that solves {\sc 3-Coloring} on graphs of maximum degree $4$ in time $2^{o(n)}$ where $n$ is the number of vertices of the input graph, which contradicts the ETH and hence completes the~proof.

The execution of {\sf ColAlg} is as follows. Given an instance $G$ of {\sc 3-Coloring} on graphs of maximum degree $4$, {\sf ColAlg} constructs the instance $\reduce(G)=(\widetilde{G},\widetilde{H},\ell)$ of {\sc LSH} in Definition \ref{def:reduceLSH} with $r=\lceil\log_{\gamma(4)}(n/\log n)\rceil$ where $n=|V(G)|$. By Lemma \ref{lem:reduceCorrecr}, $\reduce(G)=(\widetilde{G},\widetilde{H},\ell)$ is computable in  time polynomial in the sizes of $G,\widetilde{G}$ and $\widetilde{H}$, and has the following~properties: 
\begin{itemize}
\item $G$ is a \yes -instance of {\sc 3-Coloring} if and only if $(\widetilde{G},\widetilde{H},\ell)$ is a \yes -instance of {\sc LSH}.
\item $|V(\widetilde{G})|\leq n/r=\OO(n/\log n)$, and $|V(\widetilde{H})|\leq \gamma(4)^r=\OO(n/\log n)$.
\end{itemize}
Then, {\sf ColAlg} calls the polynomial-time algorithm in Lemma \ref{lem:reduceModification} with $(\widetilde{G},\widetilde{H},\ell)$ to construct an equivalent instance $(\widetilde{G},\widetilde{H}',c_{\widetilde{G}},c_{\widetilde{H}'},\ell)$ of {\sc Prop-Col LSH}, where $\widetilde{H}'$ is a subgraph of $\widetilde{H}$.
Lastly, {\sf ColAlg} calls {\sf LSHAlg} with $(\widetilde{G},\widetilde{H}',c_{\widetilde{G}},c_{\widetilde{H}'},\ell)$ as input, and returns its answer.

Since the instance $G$ of {\sc 3-Coloring} was argued above to be equivalent to the instance $(\widetilde{G},\widetilde{H}',c_{\widetilde{G}},c_{\widetilde{H}'},\ell)$ of {\sc Prop-Col LSH}, the correctness of {\sf ColAlg} directly follows. For the running time, denote $M=\max(|V(\widetilde{G})|,|V(\widetilde{H})|)$, and notice that  $M\leq \OO(n/\log n)$. Thus, because {\sf LSHAlg} runs in time $M^{o(M)} \leq (n/\log n)^{o(n/\log n)}\leq 2^{o(n)}$, it follows that {\sf ColAlg} runs in time $2^{o(n)}$. This completes the proof.
\end{proof}

In the rest of this appendix, we provide the details omitted from Section \ref{sec:listSI} regarding the transition {\sc Prop-Col LSH} to {\sc Prop-Col LSI}. We begin by adapting the Turing reduction of \cite{DBLP:journals/jacm/CyganFGKMPS17} from {\sc LSH} to {\sc LSI}.
 
 \begin{lemma}\label{lem:reduceLSHtoLSI}
 There is an $2^{\OO(n)}$-time algorithm that, given an instance $(G, H, c_G, c_H,\ell)$ of {\sc Prop-Col LSH}, returns $2^{\OO(n)}$ instances of {\sc Prop-Col LSI} having input graphs on at most $n$ vertices for $n := \max(|V(G)|,|V(H)|)$, such that $(G, H, c_G, c_H,\ell)$ is a \yes -instance of {\sc Prop-Col LSH} if and only if at least one of the returned instances is a \yes -instance of {\sc Prop-Col LSI}.
 \end{lemma}
 
 \begin{proof}
 Given an instance $(G, H, c_G, c_H, \ell)$ of  {\sc Prop-Col LSH}, the algorithm works as follows. Without loss of generality, suppose that $V(H)=\{1,2,\ldots,|V(H)|\}$. Let ${\cal P}=\{P \in \mathbb{N}_0^{|V(H)|}: \sum_{i=1}^{|V(H)|}P[i] = |V(G)|\}$. That is, $\cal P$ contains every vector with $|V(H)|$ entries that are non-negative integers whose sum is $|V(G)|$. Then, for each $P\in{\cal P}$, the algorithm returns one instance $(G,H_{P},c_{G},c_{H_P},\ell_P)$ of {\sc Prop-Col LSI} that is constructed as follows.
     \begin{itemize}
         \item The graph $H_P$ is constructed from $H$ by replacing each vertex $v\in V(H)$ with $P[v]$ copies of it, denoted $v_1,v_2 \ldots v_{P[v]}$.  (Note that $P[v]$ can be equal to $0$). Then, we connect two vertices $v_i$ to $u_j$ in $H_P$ if and only if $v$ is connected to $u$ in $H$. That is, $V(H_P)=\{v_i: v\in V(H), i\in\{1,2,\ldots,P[v]\}\}$ and $E(H_P)=\{\{u_i,v_j\}: \{u,v\}\in E(H), u_i,u_j\in V(H_P)\}$.
         \item For every vertex $u_i\in V(H_P)$, let $c_{H_P}(u_i) = c_H(u)$.
         \item For every vertex $u\in V(G)$, let $\ell_P(u)=\{v_i\in V(H_P): v\in \ell(u)\}$. 
     \end{itemize}
This completes the description of the algorithm.     
     
First, we consider some $P\in{\cal P}$ and assert that  $(G,H_{P},c_{G},c_{H_P},\ell_P)$ is indeed an instance of {\sc Prop-Col LSI}. By the construction of $V(H_P)$ and since $\sum_{i=1}^{|V(H)|}P[i] = |V(G)|$, we have that $|V(G)|=|V(H_P)|$. Clearly, as $(G, H, c_G, c_H, \ell)$ is an instance of  {\sc Prop-Col LSH}, we have that $c_G$ is a proper coloring of $G$. Now, consider an edge $\{u_i,v_j\}\in E(H_P)$. Then, $\{u,v\}\in E(H)$, and since $c_H$ is a proper coloring of $H$ (as $(G, H, c_G, c_H, \ell)$ is an instance of  {\sc Prop-Col LSH}), this means that $c_H(u)\neq c_H(v)$. By definition, $c_{H_P}(u_i)=c_{H}(u)$ and $c_{H_P}(v_i)=c_{H}(v)$, and therefore $c_{H_P}(u_i)\neq c_{H_P}(v_j)$. Thus, $c_{H_P}$ is a proper coloring of $H_P$. Lastly, consider some vertices $u\in V(G)$ and $v_i\in \ell_P(u)$. By the definition of $\ell_P$, we have that $v\in\ell(P)$. Therefore, as $(G, H, c_G, c_H, \ell)$ is an instance of  {\sc Prop-Col LSH}, $c_G(u)=c_H(v)$.  Thus, because $c_{H_P}(v_i)=c_H(v)$, we have that $c_G(u)=c_{H_P}(v_i)$.
    
Now, we consider the number of instances returned by the algorithm along with its running time. Towards this, first note that $|{\cal P}|=\binom{|V(G)|+|V(H)|-1}{|V(H)|-1} \leq 4^{n}$. As the number of returned instances equals $|{\cal P}|$, it is upper bounded by $2^{\OO(n)}$ as required. Because each instance is computed in polynomial time, we also get that the running time of the algorithm is bounded by $2^{\OO(n)}$.

Finally, we consider the correctness of the algorithm. In one direction, suppose that at least one of the returned instances is a \yes -instance of {\sc Prop-Col LSI}. Then, there exists $P\in{\cal P}$ such that $(G,H_{P},c_{G},c_{H_P},\ell_P)$ is a \yes -instance of {\sc Prop-Col LSI}. Thus, there exists a bijective function $\varphi_P: V(G)\rightarrow V(H_P)$ such that {\em (i)} for every $\{u,v\}\in E(G)$, $\{\varphi_P(u),\varphi_P(v)\}\in E(H_P)$, and {\em (ii)} for every $u\in V(G)$, $\varphi_P(u)\in \ell_P(u)$. We define a function $\varphi: V(G)\rightarrow V(H)$ as follows: for every $u\in V(G)$, let $\varphi(u)=v$ where $v\in V(H)$ is the vertex for which there exists $i\in\{1,2,\ldots,P[v]\}$ such that $\varphi_P(u)=v_i$. We now verify that $\varphi$ is a solution to the instance $(G, H,c_G, c_H, \ell)$ of  {\sc Prop-Col LSH}.
Firstly, by item {\em (i)} above, for every $\{u,v\}\in E(G)$, we have that $\{x_i,y_i\}\in E(H_P)$ where $x_i=\varphi_P(u)$ and $y_i=\varphi_P(v)$; by the definition of $H_P$, this means that $\{x,y\}\in E(H)$, and as $x=\varphi(u)$ and $y=\varphi(v)$ (by the definition of $\varphi$), we get that $\{\varphi(u),\varphi(v)\}\in E(H)$. Secondly, by item {\em (ii)} above, for every $u\in V(G)$, $v_i\in \ell_P(u)$ where $v_i=\varphi_P(u)$; by the definition of $\ell_P$, we have that $v\in\ell_P(u)$, and by the definition of $\varphi$, we have that $v=\varphi(u)$, therefore $\varphi(u)\in\ell(u)$. Thus, we conclude that $(G, H,c_G, c_H, \ell)$ is a \yes -instance of  {\sc Prop-Col LSH}.

In the other direction, suppose that $(G, H,c_G, c_H, \ell)$ is a \yes -instance of  {\sc Prop-Col LSH}. Then, there exists a function $\varphi: V(G)\rightarrow V(H)$ such that {\em (i)} for every $\{u,v\}\in E(G)$, $\{\varphi(u),\varphi(v)\}\in E(H)$, and {\em (ii)} for every $u\in V(G)$, $\varphi(u)\in \ell(u)$. Let $P$ be the vector with $|V(H)|$ entries where for each $i\in\{1,2,\ldots,|V(H)|\}$, $P[i]=|\varphi^{-1}(i)|$. Then, $\sum_{i=1}^{|V(H)|}P[i]=\sum_{i=1}^{|V(H)|}|\varphi^{-1}(i)|=|V(G)|$, and therefore $P\in{\cal P}$. Choose some arbitrary order $<$ on $V(G)$. Now, we define a function $\varphi_P: V(G)\rightarrow V(H_P)$ as follows: for every $u\in V(G)$, let $\varphi_P(u)=v_i$ where $v=\varphi(u)$ and $i=|\{w\in V(G): w\leq u, v=\varphi(w)\}|$. It should be clear that $\varphi_P$ is a bijection. Moreover, analogously to the previous direction, we assert that {\em (i)} for every $\{u,v\}\in E(G)$, $\{\varphi_P(u),\varphi_P(v)\}\in E(H_P)$, and {\em (ii)} for every $u\in V(G)$, $\varphi_P(u)\in \ell_P(u)$.  Thus, $(G,H_{P},c_{G},c_{H_P},\ell_P)$ is a \yes -instance of {\sc Prop-Col LSI}, which means that at least one of the returned instances is a \yes -instance of {\sc Prop-Col LSI}.
 \end{proof}

We are ready to complete the proof of Lemma \ref{res:PropColLSI}.

\begin{proof}[Proof of Lemma \ref{res:PropColLSI}]
Targeting a contradiction, suppose that there exists an algorithm, denoted by {\sf LSIAlg}, that solves {\sc Prop-Col LSI} in time $n^{o(n)}$ where where $n = \max(|V(G)|,|V(H)|)$ for input graphs $G$ and $H$. We will show that this implies the existence of an algorithm, denoted by {\sf LSHAlg}, that solves {\sc Prop-Col LSH} in time $n^{o(n)}$ where $n = \max(|V(G)|,|V(H)|)$ for input graphs $G$ and $H$, which contradicts Lemma \ref{res:PropColLSIH} and hence completes the proof.

The execution of {\sf LSHAlg} is as follows. Given an instance $(G,H,c_H,c_G,\ell)$ of {\sc Prop-Col LSH}, {\sf LSHAlg} calls the algorithm in Lemma \ref{lem:reduceLSHtoLSI} so that in time $2^{\OO(n)}$ it obtains $2^{\OO(n)}$ instances of {\sc Prop-Col LSI} having input graphs on at most $n$ vertices for $n := \max(|V(G)|,|V(H)|)$, such that $(G, H, c_G, c_H,\ell)$ is a \yes -instance of {\sc Prop-Col LSH} if and only if at least one of the returned instances is a \yes -instance of {\sc Prop-Col LSI}. Then, it calls {\sf LSIAlg} on each of the returned instances, and returns \yes\ if and only if at least one of these calls returns \yes. It should be clear that {\sf LSHAlg} runs in time $n^{o(n)}$ and that it is correct.
\end{proof}

\section{Lower Bounds for Contraction to Graph Classes Problems}\label{sec:contractionToClasses}

In this section, we prove lower bounds for several cases of the {\sc $\cal F$-Contraction} problem, defined as follows. Here, $\cal F$ is a (possibly infinite) family of graphs.

\defproblem{{\sc $\cal F$-Contraction}}{A graph $G$ and $t\in\mathbb{N}$.}{Does there exist a subset $F\subseteq E(G)$ of size at most $t$ such that $G/F\in {\cal F}$?}

Notice that {\sc Clique Contraction} is the case of {\sc $\cal F$-Contraction} where $\cal F$ is the family of cliques. In this section, we consider the cases of {\sc $\cal F$-Contraction} where $\cal F$ is the family of chordal graphs, interval graphs, proper interval graphs, threshold graphs, trivially perfect graphs, split graphs, complete split graphs and perfect graphs, also called  {\sc Chordal Contraction}, {\sc Interval Contraction}, {\sc Proper Interval Contraction}, {\sc Threshold Contraction}, {\sc Trivially Perfect Contraction}, {\sc Split Contraction}, {\sc Complete Split Contraction} and {\sc Perfect Contraction}, respectively. Before we define these classes formally, it will be more enlightening to first define only the class of chordal graphs as well as somewhat artificial classes of graphs that will help us prove lower bounds for many of the classes above in a unified manner.

\begin{definition}[{\bf Chordal Graphs}]
A  graph is {\em chordal} if it does not contain $C_\ell$ for all $\ell\geq 4$ as an induced subgraph.
\end{definition}

Our first class of graphs is defined as follows (see Fig.~\ref{fig:twoCliques}).

\begin{figure}[t]
 \center{\includegraphics[scale=0.6]{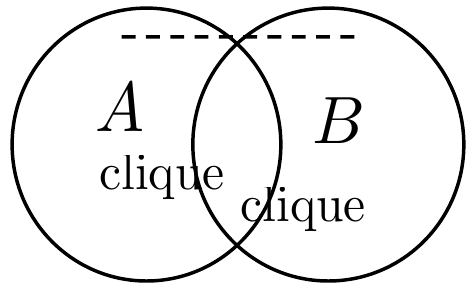}}
\caption{\label{fig:twoCliques} A two-cliques graph (see Definition \ref{def:twoCliues}).}
\end{figure}

\begin{definition}[{\bf Two-Cliques Graphs}]\label{def:twoCliues}
A {\em two-cliques graph} is a graph $G$ such that there exist $A,B\subseteq V(G)$ such that $A\cup B=V(G)$, $G[A]$ and $G[B]$ are cliques, and there do not exist vertices $a\in A\setminus B$ and $b\in B\setminus A$ such that $\{a,b\}\in E(G)$. 
The {\em two-cliques class} is the class of all two-cliques graphs.
\end{definition}

It should be clear that the two-cliques class is a subclass of the class of chordal graphs. Now, we further define a family of classes of graphs as follows.

\begin{definition}[{\bf Non-Trivial Chordal Class}]
We say that a class of graphs $\cal F$ is {\em non-trivial chordal} if it is a subclass of the class of chordal graphs, and a superclass of the two-cliques class.
\end{definition}

Clearly, the class of cliques is not a non-trivial chordal class, and the class of chordal graphs is a non-trivial chordal class. The rest of this section is divided as follows. First, in Section \ref{section:nonTrivChordal}, we prove a lower bound for any non-trivial chordal class. Then, in Section \ref{section:otherClasses}, we prove a lower bound for some graph classes that are not non-trivial chordal.

\subsection{Non-Trivial Chordal Graph Classes}\label{section:nonTrivChordal}

The main objective of this subsection is to prove the following theorem. Afterwards, we will derive lower bounds for several known graph classes as corollaries.

\begin{theorem}\label{thm:nonTrivChordal}
Let $\cal F$ be any non-trivial chordal graph class. Unless the ETH is false, there does not exist an algorithm that solves {\sc $\cal F$-Contraction} in time $n^{o(n)}$ where $n=|V(G)|$.
\end{theorem}

For the proof of this theorem, the following well-known property of chordal graphs will come in handy. This property is a direct consequence of the alternative characterization of the class of chordal graphs as the class of graphs that admit clique-tree decompositions, see~\cite{Golumbic:2004:AGT:984029}.

\begin{proposition}\label{prop:chordal}
Let $G$ be a chordal graph, and let $u$ and $v$ be two non-adjacent vertices in $G$. Then, $G[N(u)\cap N(v)]$ is a clique.
\end{proposition}

We are now ready to prove Theorem \ref{thm:nonTrivChordal}.

\begin{proof}[Proof of Theorem \ref{thm:nonTrivChordal}.]
Targeting a contradiction, suppose that there exists an algorithm, denoted by {\sf NonTrivChordAlg}, that solves {\sc $\cal F$-Contraction} in time $n^{o(n)}$ where $n$ is the number of vertices in the input graph. We will show that this implies the existence of an algorithm, denoted by {\sf CliConAlg}, that solves {\sc Structured Clique Contraction} in time $n^{o(n)}$ where $n$ is the number of vertices in the input graph, thereby contradicting Lemma~\ref{res:StructCliqueContraction} and hence completing the proof.

\begin{figure}[t]
 \center{\includegraphics[scale=0.6]{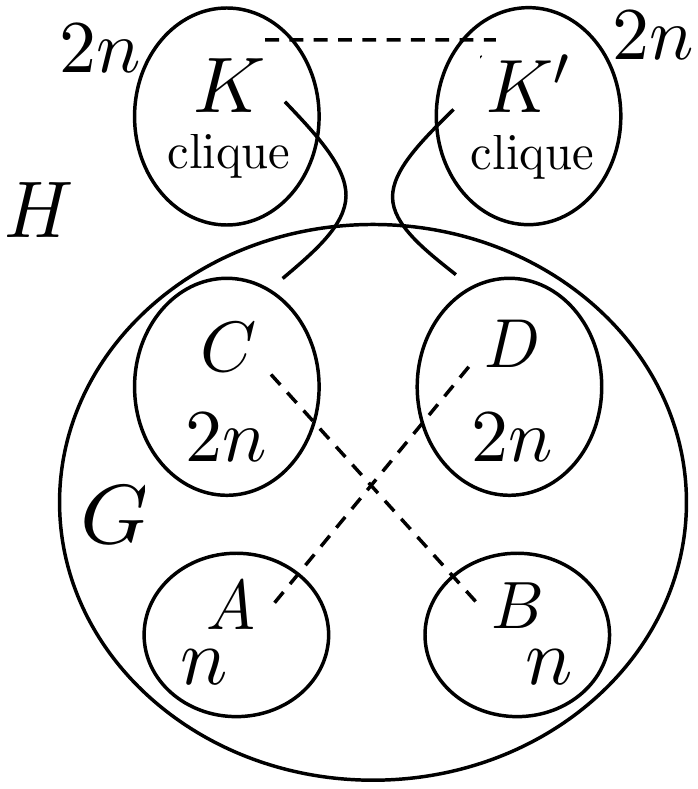}}
\caption{\label{fig:nonTrivChord} The construction of an instance of {\sc $\cal F$-Contraction} in the proof of Theorem~\ref{thm:nonTrivChordal} where dashed lines represent non-edges.}
\end{figure}

We define the execution of {\sf CliConAlg} as follows. Given an instance $(G,A,B,C,D,n)$ of {\sc Structured Clique Contraction}, {\sf CliConAlg} constructs an instance $(H,n)$ of {\sc $\cal F$-Contraction} as follows (see Fig.~\ref{fig:nonTrivChord}):
\begin{itemize}
\item Let $n=|A|$. Moreover, let $K$ and $K'$ be two cliques, each on $2n$ new vertices.
\item $V(H)=V(G)\cup V(K)\cup V(K')$.
\item $E(H)=E(G)\cup E(K)\cup E(K')\cup \{\{u,v\}: u\in V(G), v\in V(K)\cup V(K')\}$.
\end{itemize}
Then, {\sf CliConAlg} calls {\sf NonTrivChordAlg} with $(H,n)$ as input, and returns the answer of this call.

First, note that by construction, $|V(H)|=10n$. Thus, because {\sf NonTrivChordAlg} runs in time $|V(H)|^{o(|V(H)|)} \leq n^{o(n)}$, it follows that {\sf CliConAlg} runs in time $n^{o(n)}$.

For the correctness of the algorithm, first suppose that $(G,A,B,C,D,n)$ is a \yes -instance of {\sc Structured Clique Contraction}. This means that there exists a subset $F\subseteq E(G)$ of size at most $n$ such that $G/F$ is a clique. By the definition of $H$, we directly derive that $H/F$ is a two-cliques graphs, and therefore it belongs to $\cal F$.  Thus, $(H,n)$ is a \yes -instance of {\sc $\cal F$-Contraction}, which means that the call to {\sf NonTrivChordAlg} with $(H,n)$ as input returns \yes, and hence {\sf CliConAlg} returns \yes. 

Now, suppose that {\sf CliConAlg} returns \yes, which means that the call to {\sf NonTrivChordAlg} with $(H,n)$ returns \yes.  Thus, $(H,n)$ is a \yes -instance of {\sc $\cal F$-Contraction}, which means that there exists a subset $F\subseteq E(H)$ of size at most $n$ such that $H/F\in {\cal F}$. In particular, $H/F$ is a chordal graph. Based on Proposition~\ref{prop:chordal}, we will first show that $H[A\cup B\cup C\cup D\cup X]/F$ is a clique, where $X=\{u\in V(K)\cup V(K'):$ there exists a vertex $v\in A\cup B\cup C\cup D$ such that $u$ and $v$ belong to the same connected component of $H[F]\}$.

Targeting a contradiction, suppose that $H[A\cup B\cup C\cup D\cup X]/F$ is not a clique, and therefore there exist two non-adjacent vertices $u$ and $v$ in this graph. By the definition of $X$, $H[A\cup B\cup C\cup D\cup X]/F$ is equal to the subgraph of $H/F$ induced by the set of vertices derived from connected components that contain at least one vertex from $A\cup B\cup C\cup D$. In particular, $u$ and $v$ are also non-adjacent vertices in $H/F$. By Proposition~\ref{prop:chordal}, this implies that $(H/F)[N_{H/F}(u)\cap N_{H/F}(v)]$ is a clique. Let ${\cal C}_1$ (resp.~${\cal C}_2$) be the set of connected components of $H[F]$ that contain at least one vertex from $V(K_1)$ (resp.~$V(K_2)$). Because $|F|\leq n$ and $|V(K_1)|=|V(K_2)|=2n$, there exists at least one component $C_1\in{\cal C}_1$ (resp.~$C_2\in {\cal C}_2$) that does not contain any vertex from $A\cup B\cup C\cup D$. Let $c_1$ and $c_2$ be the vertices of $H/F$ yielded by the replacement of $C_1$ and $C_2$, respectively. As all vertices in $V(K_1)\cup V(K_2)$ are adjacent to all vertices in $A\cup B\cup C\cup D$, we have that $c_1,c_2\in N_{H/F}(u)\cap N_{H/F}(v)$. However, there do not exist a vertex in $V(K_1)$ and a vertex in $V(K_2)$ that are adjacent in $H$, and for every vertex in $V(K_1)\cup V(K_2)$, its neighborhood outside this set is contained in $A\cup B\cup C\cup D$. Thus, $c_1$ and $c_2$ must be non-adjacent in $H/F$. However, this is a contradiction to the argument that $(H/F)[N_{H/F}(u)\cap N_{H/F}(v)]$ is a clique. From this, we derive that $H[A\cup B\cup C\cup D\cup X]/F$ is indeed a clique.

Now, notice that $(H,A,B,C,D,N,n)$ where $N=V(K_1)\cup V(K_2)$ is an instance of {\sc Noisy Structured Clique Contraction}. Furthermore, since $|F|\leq n$ and we have already shown that $H[A\cup B\cup C\cup D\cup X]/F$ is a clique, we have that $F$ is a solution to this instance. Therefore, by Lemma~\ref{lem:propNoisyStructured}, $F$ is a matching of size $n$ in $H$ such that each edge in $F$ has one endpoint in $A$ and the other in $B$. In particular, $F\subseteq E(G)$ and hence $X=\emptyset$. Because $G=H[A\cup B\cup C\cup D]$, we thus derive that $G/F$ is a clique. Thus, we conclude that $(G,A,B,C,D,n)$ is a \yes -instance of {\sc Structured Clique Contraction}. This completes the proof of the reverse direction.
\end{proof}

Now, we give definitions for several classes of graphs for which lower bounds will follow from Theorem \ref{prop:chordal}. First, a graph is an {\em interval graph} if there exists a set of intervals on the real line such that the vertices of the graph are in bijection with these intervals,  and there exists edge between two vertices if and only if their intervals intersect. A graph is a  {\em proper interval graph} if, in the former definition, we also add the constraint that all intervals must have the same length. A graph is a {\em threshold graph} if it can be constructed from a one-vertex graph by repeated applications of the following two operations: addition of a single isolated vertex to the graph; addition of a single vertex that is connected to all other vertices. A graph is {\em trivially perfect} if in each of its induced subgraphs, the maximum size of an independent set equals the number of maximal cliques.

It is well-known that every graph that is a (proper) interval graph, or a threshold graph, or a trivially perfect graph, is also a chordal graph (see \cite{Golumbic:2004:AGT:984029}). Moreover, it is immediate to verify that the two-cliques class is a subclass of the classes of (proper) interval graphs, threshold graphs and trivially perfect graphs. Thus, these classes are non-trivial chordal graphs classes, and therefore Theorem \ref{thm:nonTrivChordal} directly implies lower bounds for them as state below. 

\begin{corollary}\label{cor:graphClasses}
Unless the ETH is false, none of the following problems admits an algorithm that solves it in time $n^{o(n)}$ where $n=|V(G)|$: {\sc Chordal Contraction}, {\sc Interval Contraction}, {\sc Proper Interval Contraction}, {\sc Threshold Contraction} and {\sc Trivially Perfect Contraction}.
\end{corollary}

\subsection{Other Graph Classes}\label{section:otherClasses}

In Section \ref{sec:hadwiger}, we have already proved a lower bound for a class of graphs that is not non-trivial chordal, namely, the class of cliques. In this section, we show that our approach can yield lower bounds for other classes of graphs that are not non-trivially chordal. For illustrative purposes, we consider the classes of {\sc Split Graphs}, {\sc Complete Split Graphs} and {\sc Perfect Graphs}.

A graph $G$ is a {\em split graph} if there exists a partition $(I,K)$ of $V(G)$ such that $G[I]$ is edgeless and $G[K]$ is a clique. In case $\{\{i,k\}: i\in I, k\in K\}$, we further say that $G$ is a {\em complete split graph}. Notice  that the two-cliques class is not a subclass of the class of split graphs, and hence the class of (complete) split graphs is not non-trivially chordal.

For the class of (complete) split graphs, we prove the following statement.

\begin{theorem}\label{thm:SplitContraction}
Unless the ETH is false, there does not exist an algorithm that solves {\sc Split Contraction} (or {\sc Complete Split Contraction}) in time $n^{o(n)}$ where $n=|V(G)|$.
\end{theorem}

\begin{proof}
Targeting a contradiction, suppose that there exists an algorithm, denoted by {\sf SplitAlg}, that solves {\sc Split Contraction} (or {\sc Complete Split Contraction}) in time $n^{o(n)}$ where $n$ is the number of vertices in the input graph. We will show that this implies the existence of an algorithm, denoted by {\sf CliConAlg}, that solves {\sc Structured Clique Contraction} in time $n^{o(n)}$ where $n$ is the number of vertices in the input graph, thereby contradicting Lemma~\ref{res:StructCliqueContraction} and hence completing the proof.

\begin{figure}[t]
 \center{\includegraphics[scale=0.6]{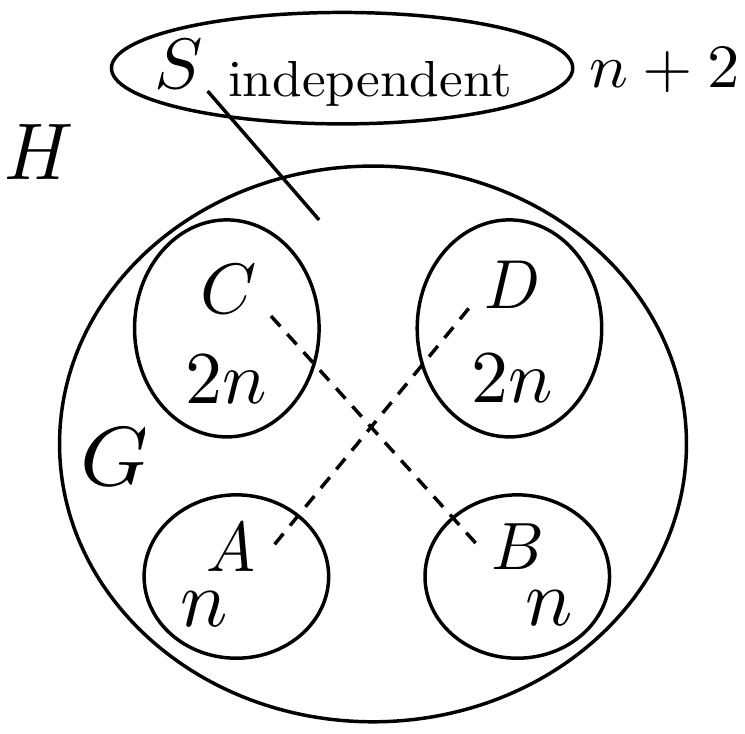}}
\caption{\label{fig:split} The construction of an instance of {\sc Split Contraction} in the proof of Theorem~\ref{thm:SplitContraction} where dashed lines represent non-edges.}
\end{figure}

We define the execution of {\sf CliConAlg} as follows. Given an instance $(G,A,B,C,D,n)$ of {\sc Structured Clique Contraction}, {\sf CliConAlg} constructs an instance $(H,n)$ of {\sc Split Contraction} (or {\sc Complete Split Contraction}) as follows (see Fig.~\ref{fig:split}):
\begin{itemize}
\item $V(H)=V(G)\cup S$ where $S$ is a set of $n+2$ new vertices.
\item $E(H)=E(G)\cup \{\{u,v\}: u\in V(G), v\in S\}$.
\end{itemize}
Then, {\sf CliConAlg} calls {\sf SplitAlg} with $(H,n)$ as input, and returns the answer of this call.

First, note that by construction, $|V(H)|=7n+2$. Thus, because {\sf SplitAlg} runs in time $|V(H)|^{o(|V(H)|)} \leq n^{o(n)}$, it follows that {\sf CliConAlg} runs in time $n^{o(n)}$.

For the correctness of the algorithm, first suppose that $(G,A,B,C,D,n)$ is a \yes -instance of {\sc Structured Clique Contraction}. This means that there exists a subset $F\subseteq E(G)$ of size at most $n$ such that $G/F$ is a clique. By the definition of $H$, we derive that $H/F$ is a complete split graph: $(S,V(G/F))$ is a partition of $V(H/F)$ where $S$ induces an independent set, $V(G/F)$ induces a clique, and every vertex in $S$ is adjacent to every vertex in $V(G/F)$.  Thus, $(H,n)$ is a \yes -instance of {\sc Complete Split Contraction} (as well as of {\sc Split Contraction}), which means that the call to {\sf SplitAlg} with $(H,n)$ returns \yes, and hence {\sf CliConAlg} returns \yes. 

Now, suppose that {\sf CliConAlg} returns \yes, which means that the call to {\sf SplitAlg} with $(H,n)$ returns \yes.  Thus, $(H,n)$ is a \yes -instance of {\sc Split Contraction} (even if {\sf SplitAlg} solves {\sc Complete Split Contraction}), which means that there exists a subset $F\subseteq E(H)$ of size at most $n$ such that $H/F$ is a split graph. Let $(I,K)$ be a partition of $V(H/F)$ into an independent set and a set of vertices that induce a clique.  Because $|S|=n+2$ and $H[S]$ is an independent set, there exist at least two vertices $s_1,s_2\in S$ that are not incident to any edge in $F$. As these vertices are not adjacent to one another in $H$, and because they are adjacent to all vertices in $V(G)$ (and hence to all vertices in $V(H/F)\setminus S$), it follows that $s_1,s_2\in I$ and $V(H/F)\setminus S\subseteq K$. In particular, $(H/F)[V(H/F)\setminus S]$ is a clique. Let $X=\{u\in S:$ there exists a vertex $v\in V(G)$ such that $u$ and $v$ belong to the same connected component of $G[F]\}$. Then, we have that $H[V(G)\cup X]/F$ is a clique.

Now, notice that $(H,A,B,C,D,S,n)$ is an instance of {\sc Noisy Structured Clique Contraction}. Furthermore, since $|F|\leq n$ and we have already shown that $H[A\cup B\cup C\cup D\cup X]/F$ is a clique, we have that $F$ is a solution to this instance. Therefore, by Lemma~\ref{lem:propNoisyStructured}, $F$ is a matching of size $n$ in $H$ such that each edge in $F$ has one endpoint in $A$ and the other in $B$. In particular, $F\subseteq E(G)$ and hence $X=\emptyset$. Because $G=H[A\cup B\cup C\cup D]$, we thus derive that $G/F$ is a clique. Thus, we conclude that $(G,A,B,C,D,n)$ is a \yes -instance of {\sc Structured Clique Contraction}. This completes the proof of the reverse direction.
\end{proof}

A graph $G$ is a {\em perfect graph} if the chromatic number of every induced subgraph of $G$ equals the size of the largest clique of that subgraph. Here, the chromatic number of a graph is the minimum number of colors required to color its vertices so that every pair of adjacent vertices are assigned different colors. For the class of perfect graphs, we prove the following statement.

\begin{theorem}\label{thm:PerfectContraction}
Unless the ETH is false, there does not exist an algorithm that solves {\sc Perfect Contraction} in time $n^{o(n)}$ where $n=|V(G)|$.
\end{theorem}

\begin{proof}
Targeting a contradiction, suppose that there exists an algorithm, denoted by {\sf PerfectAlg}, that solves {\sc PerfectContraction}  in time $n^{o(n)}$ where $n$ is the number of vertices in the input graph. We will show that this implies the existence of an algorithm, denoted by {\sf CliConAlg}, that solves {\sc Structured Clique Contraction} in time $n^{o(n)}$ where $n$ is the number of vertices in the input graph, thereby contradicting Lemma~\ref{res:StructCliqueContraction} and hence completing the proof.

\begin{figure}[t]
 \center{\includegraphics[scale=0.6]{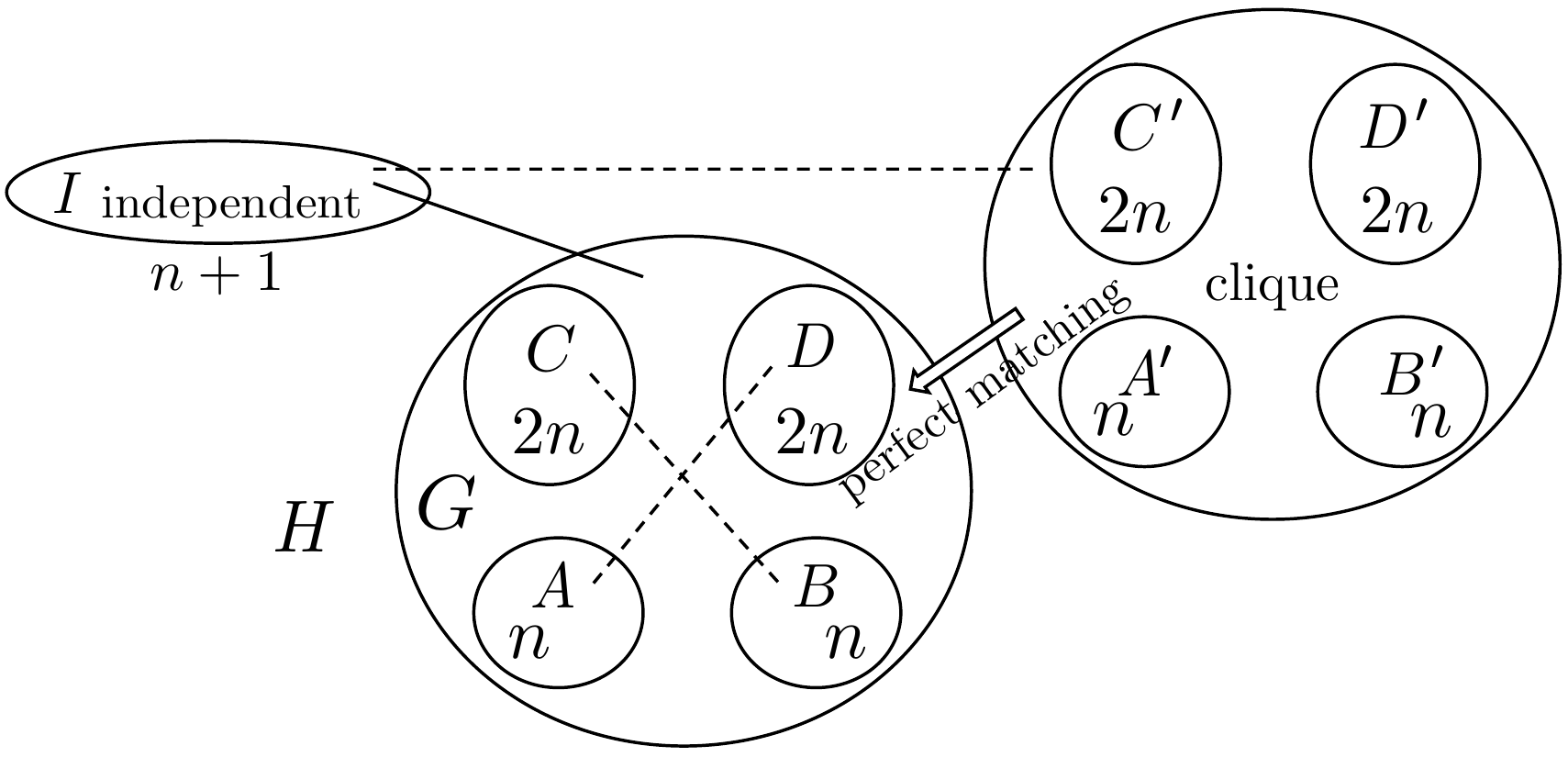}}
\caption{\label{fig:perfect} The construction of an instance of {\sc Perfect Contraction} in the proof of Theorem~\ref{thm:PerfectContraction} where dashed lines represent non-edges.}
\end{figure}

We define the execution of {\sf CliConAlg} as follows. Given an instance $(G,A,B,C,D,n)$ of {\sc Structured Clique Contraction}, {\sf CliConAlg} constructs an instance $(H,n)$ of {\sc Perfect Contraction}  as follows (see Fig.~\ref{fig:perfect}):
\begin{itemize}
\item Let $K=\{u': u\in V(G)\}$ where each element $u'$ is a new vertex referred to as the {\em tagged copy of $u$}. Additionally, let $I$ be a set of $n+1$ new vertices.
\item $V(H)=V(G)\cup K\cup I$.
\item $E(H)=E(G)\cup  \{\{u,u'\}: u\in V(G)\}\cup \{\{u',v'\}: u',v'\in K\}\cup\{\{u,i\}: u\in V(G), i\in I\}$.
\end{itemize}
Then, {\sf CliConAlg} calls {\sf PerfectAlg} with $(H,n)$ as input, and returns the answer of this call.

First, note that by construction, $|V(H)| \leq 13n+1$. Thus, because {\sf PerfectAlg} runs in time $|V(H)|^{o(|V(H)|)} \leq n^{o(n)}$, it follows that {\sf CliConAlg} runs in time $n^{o(n)}$.

In what follows, given a subset $U\subseteq V(G)$, we denote $U'=\{u'\in K: u\in U\}$. For the correctness of the algorithm, first suppose that $(G,A,B,C,D,n)$ is a \yes -instance of {\sc Structured Clique Contraction}. This means that there exists a subset $F\subseteq E(G)$ of size at most $n$ such that $G/F$ is a clique. Now, we will show that $H/F$ is a perfect graph. To this end, consider some induced subgraph $S$ of $H/F$. In case the maximum size of a clique in $S$ is $2$, then $S$ can contain at most four non-leaf vertices: at most two vertices from $K$ and at most two vertices from outside $K\cup I$ (because $H[V(G)]/F$ is a clique); then, it is trivial to color $S$ with number of colors equal to its maximum clique size---in fact, it is straightforward to verify that any graph on at most four vertices is perfect. Thus, in what follows, suppose that the maximum size of a clique in $S$ is at least $3$. Now, consider a clique $\widehat{C}$ of maximum size in $S$, and observe that it must either consist only of vertices in $K$ or of no vertex in $K$ (in which case it can contain at most one vertex from $I$). In the first case, color each vertex in $u'\in V(\widehat{C})$ by a distinct color, and note that all vertices in $V(S)\setminus V(\widehat{C})$ can be colored using the same set of colors so that a vertex and its tagged copy are assigned distinct colors. The second case is analogous. In either case, we obtain that the chromatic number of $S$ equals its maximum clique size. Thus, $(H,n)$ is a \yes -instance of {\sc Perfect Contraction}, which means that the call to {\sf PerfectAlg} with $(H,n)$ returns \yes, and hence {\sf CliConAlg} returns \yes. 

Now, suppose that {\sf CliConAlg} returns \yes, which means that the call to {\sf PerfectAlg} with $(H,n)$ returns \yes.  Thus, $(H,n)$ is a \yes -instance of {\sc Perfect Contraction}, which means that there exists a subset $F\subseteq E(H)$ of size at most $n$ such that $H/F$ is a perfect graph. We first argue that there does not exist a vertex $a\in A\cup B$ such that neither $a$ nor $a'$ is incident to at least one edge in $F$. Targeting a contradiction, suppose that there exists $a\in A\cup B$ such that neither $a$ not $a'$ is incident to at least one edge in $F$. Assume that $a\in A$ as the other case is symmetric. Because $|F|\leq n$ and $|D|=2n$, there either exists a vertex $d\in D$ such that neither $d$ nor $d'$ is incident to at least one edge in $F$, or $F$ is a perfect matching in either $G[D]$ or $G[D']$, where in the latter case we let $d$ denote some arbitrarily chosen vertex from $D$. Additionally, since $I$ is an independent set of size $n+1$, there exists a vertex $i\in I$ that is not incident to any edge in $F$. Now, consider the cycle $i-a-a'-d'-d-i$ (on five vertices)  in $H$. This cycle is an induced cycle in $H$, because no vertex in $A$ is adjacent to any vertex in $D$, and by the construction of $H$, $i$ is not adjacent to $a'$ and $d'$, $a$ is not adjacent to $d'$ and $a'$ is not adjacent to $d$. Furthermore, as $i,a$ and $a'$ are not incident to any edge in $F$, and if any of $d$ and $d'$ is incident to an edge in $F$, then $F$ is a perfect matching in either $G[D]$ or $G[D']$, we obtain that $i-a-a'-\widehat{d}-\widehat{d}'-i$ is an induced cycle (on five vertices)  in $H/F$, where $\widehat{d}$ and $\widehat{d}'$ are the vertices yielded by the replacement of the connected components of $H[F]$ that contain $d$ and $d'$, respectively, if such components exist (otherwise, $\widehat{d}=d$ and $\widehat{d}'=d'$). However, an induced cycle on five vertices has chromatic number $3$ and maximum clique size $2$, thus we derive a contradiction to the supposition that $H/F$ is perfect.

So far, we derived that there does not exist a vertex $a\in A\cup B$ such that neither $a$ nor $a'$ is incident to at least one edge in $F$. As $|F|\leq n$ and $|A|=|A'|=|B|=|B'|=n$, this means that every edge in $F$ has both endpoints in $A\cup A'\cup B\cup B'$ and that for each $u\in A\cup B$, exactly one vertex among $u$ and $u'$ is incident to an edge in $F$. Now, we will show that each vertex $a\in A\cup B$ is incident to at least one edge in $F$. Targeting a contradiction, suppose that there exists a vertex $a\in A\cup B$ that is not incident to any edge in $F$. Assume that $a\in A$, as the other case is symmetric. Denote $i,d,d',\widehat{d}$ and $\widehat{d}'$ as before, and again consider the induced cycle $i-a-a'-d'-d-i$ in $H$. Unlike before, now $a'$ belongs to some connected component of $H[F]$, yet we know that this connected component consists only of $a'$ and some vertex in $B'$. Let $\widehat{a}'$ be the vertex yielded by the replacement of this component. As no vertex in $B'$ is adjacent to any vertex in $I\cup D$, we again have that  $i-a-\widehat{a}'-\widehat{d}'-\widehat{d}-i$ is an induced cycle in $H/F$, which gives rise to a contradiction. Thus, as $|F|\leq n$ and $|A|=|B|=n$, we know that $F$ is a perfect matching in $G[A\cup B]$.

Next, we will show that $G/F$ is a clique. This will imply that $(G,A,B,C,D,n)$ is a \yes -instance of {\sc Structured Clique Contraction} and thereby complete the proof. Targeting a contradiction, suppose that $G/F$ is not a clique, and therefore there exist two non-adjacent vertices $u$ and $v$ in $G/F$. As $F$ is a matching in $G[A\cup B]$, we can let $x$ and $y$ be two vertices in $A\cup B$ that belonged to the connected components of $H[F]$ that yielded $u$ and $v$, respectively.  Notice that the only vertex in $A\cup B$  adjacent to $x'$ is $x$, and the analogous claim holds for $y'$ and $y$. As $F$ is a matching in $G[A\cup B]$ that does not match $x$ and $y$ (since otherwise $u$ and $v$ would not be distinct vertices), we have that neither  $u$ is adjacent to $y'$ in $H/F$ nor $v$ is adjacent to $x'$ in $H/F$. From this, by the construction of $H$ and since $F$ is a matching in $G[A\cup B]$, we immediately derive that $i-u-x'-y'-v-i$ is an induced cycle in $H/F$ where $i$ is some arbitrarily chosen vertex from $I$. However, as before, the existence of such a cycle contradicts the supposition that $H/F$ is a perfect graph. Thus, the proof of the reverse direction is complete.
\end{proof}

\end{document}